\documentclass[journal,onecolumn,10pt,draftcls]{IEEEtran}
\usepackage{enumerate}
\ifCLASSOPTIONcompsoc
  \usepackage[caption=false,font=normalsize,labelfont=sf,textfont=sf]{subfig}
\else
  \usepackage[caption=false,font=footnotesize]{subfig}
\fi

\usepackage{url}
\usepackage{cite}
\usepackage{enumitem}
\usepackage[colorlinks=true,linkcolor=red,citecolor=blue]{hyperref}%
\usepackage[cmex10]{amsmath} \usepackage{bbm}

\usepackage{todonotes}

\usepackage{mathtools}
\usepackage{subfig}
\usepackage{amsthm}
\usepackage{tikz}
\usepackage{soul}
\usepackage{algorithm}
\usepackage{bm}
\usepackage{algpseudocode}
\usepackage{amssymb}
\usepackage{amsmath}
\usepackage{url}
\usepackage{graphicx}
\usepackage{xcolor}
\let\emptyset\varnothing
\newcommand{\RN}[1]{\textup{\uppercase\expandafter{\romannumeral#1}}}

\newtheorem{theorem}{Theorem}
\theoremstyle{definition}

\newtheorem{lm}{Lemma}
\newtheorem{defi}{Definition}
\newtheorem{proposition}{Proposition}
\newtheorem{corollary}{Corollary}

\newtheorem{remark}{Remark}
\usepackage{thmtools}
\usepackage{url}
\DeclareTextFontCommand{\textroman}{\fontlibertine}

\usepackage{amsmath}

\theoremstyle{definition}
\newtheorem{example}{Example}

\makeatletter
\newcommand*{\rom}[1]{\expandafter\@slowromancap\romannumeral #1@}
\newcommand{\w}{\omega}
\newcommand{\s}{\star}
\newcommand{\fs}{f^\star}
\newcommand{\cI}{\mathcal{I}}
\newcommand{\cS}{\mathcal{S}}
\newcommand{\cT}{\mathcal{T}}

\newcommand{\cL}{\mathcal{L}}
\newcommand{\n}{\chi}
\newcommand{\nB}{B} 
\renewcommand{\b}{b} 
\newcommand{\ml}{m} 

\newcommand{\bfw}{\boldsymbol{\omega}}

\newcommand{\bfs}{\boldsymbol{\sigma}}
\newcommand{\cache}{\mathfrak{C}}
\newcommand{\bc}{\mathbf{c}}
\newcommand{\nn}{N}

\newcommand{\bft}{\boldsymbol{\theta}}

\newcommand{\bx}{\mathbf{x}}
\newcommand{\bA}{\mathbf{A}}

\newcommand{\bG}{\mathbf{G}}
\newcommand{\bH}{\mathbf{H}}

\newcommand{\bb}{\mathbf{b}}
\newcommand{\bz}{\mathbf{z}}

\newcommand{\bd}{\boldsymbol{d}}

\newcommand{\fl}{f_{\mathsf{LP}}}

\renewcommand{\P}{\mathbb{P}}

\newcommand{\lp}{\left(}
\newcommand{\rp}{\right)}
\newcommand{\lb}{\left[}
\newcommand{\rb}{\right]}

\newcommand{\lc}{\left\{}
\newcommand{\rc}{\right\}}

\newcommand{\ov}[1]{\overline{ #1 }}
\newcommand{\wdt}[1]{\widetilde{ #1 }}

\newcommand{\by}{\mathbf{y}}

\newcommand{\comment}[1]{}

\newcommand{\cQ}{\mathcal{Q}}

\newcommand{\zeros}{\mathbf{0}}
\newcommand{\ones}{\mathbf{1}}

\newcommand{\pis}{\pi^\s}

\newcommand{\bof}{\boldsymbol{\ov{F}}}
\newcommand{\bfI}{\boldsymbol{I}}
\newcommand{\cent}{\cache^{\mathsf{cent}}}

\newcommand{\ccent}[1]{\bc^{\mathsf{cent}}_{#1}}

\newcommand{\tcent}{\bmu^{\mathsf{cent}}}
\newcommand{\tcents}{\mu^{\mathsf{cent}}}
\newcommand{\htf}{\hat{f}}

\newcommand{\bmu}{\boldsymbol{\mu}} 
\newcommand{\md}{\middle | }

\newcommand{\kw}{k_{\mathsf{w}}(\cS)}
\newcommand{\el}{\ell}

\begin{document}

\title{Cache-Aided $K$-User Broadcast Channels\\ with State Information at Receivers}

\author{%
  \IEEEauthorblockN{%
    Hadi~Reisizadeh, 
    Mohammad Ali Maddah-Ali, and  Soheil~Mohajer%
    \thanks{The work of H.~Reisizadeh and S.~Mohajer was supported in part by the National Science Foundation under Grants CCF-1749981. A preliminary version of this work was presented in part at the 2019 IEEE International Symposium on Information Theory~\cite{reisizadeh2019cache}.}
    \thanks{The authors are with the Department of Electrical and Computer Engineering, University of Minnesota, Minneapolis, MN, 55455 USA (e-mail: hadir@umn.edu; maddah@umn.edu; soheil@umn.edu).   Corresponding author: S.~Mohajer.}
  }%
  }
\maketitle
\date{}

\begin{abstract} 
We study a $K$-user coded-caching broadcast problem in a joint source-channel coding framework. The transmitter observes a database of files that are being generated at a certain rate per channel use, and each user has a cache, which can store a fixed fraction of the generated symbols. In the delivery phase, the transmitter broadcasts a message so that the users can decode their desired files using the received signal and their cache content. The communication between the transmitter and the receivers happens over a (deterministic) \emph{time-varying} erasure broadcast channel, and the channel state information is only available to the users. We characterize the maximum achievable source rate for the $2$-user and the degraded $K$-user problems. We provide an upper bound for any caching strategy's achievable source rates. Finally, we present a linear programming formulation to show that the upper bound is not a sharp characterization. Closing the gap between the achievable rate and the optimum rate remains open.  
\end{abstract}

\begin{IEEEkeywords}
Coded caching, joint source-channel coding, broadcast channel, wireless networks.
\end{IEEEkeywords}

\section{Introduction}
\IEEEPARstart{T}{he} number of active users of video streaming applications such as Netflix, YouTube, HBO, etc., is growing rapidly. Coded caching is a promising strategy to overcome this rapidly growing traffic load of networks during their peak traffic time by duplicating parts of the content in the caches distributed across the network. A caching system operates in two phases: (i) a placement (pre-fetching) phase, where each user has access to the database of the transmitter and stores some packets from the database during the off-load time, and (ii) a delivery (fetching) phase, during which each user demands a file from the database, and the transmitter broadcasts a signal over a (noisy) channel to all the users (receivers), such that each user is able to decode his desired file from his cache content and his received signal. Moreover, in this phase, the network is congested, and the transmitter exploits the content of users to serve their requested files.  

In practice, assuming a perfect broadcast channel fails, especially for the wireless communication setup. Therefore, we are dealing with a random time-varying channel between the transmitter and the users. In this paper, to model the randomness of the channel, we consider a binary deterministic version of a time-varying memoryless fading broadcast channel when the transmitter is serving users. However, the pre-fetching phase takes over the noiseless links. We study such a caching problem in a joint source-channel coding framework and analyze the limitations of the source rate for the transmitter. 

\noindent\textbf{Related Works.}
Coded caching schemes are proposed under the perfect channel assumption for the delivery phase and the uncoded cache placement where the placement performs on pure packets of the files for the centralized ~\cite{maddah2014fundamental} and the decentralized settings~\cite{maddah2014decentralized}. It is shown that a significant gain can be achieved by sending coded packets and simultaneously serving multiple users. A distributed source coding problem is presented in~\cite{lim2017information} to study the cache-aided networks. The database is viewed as a discrete memoryless source and the users' requests as side information that is available everywhere except at the cache encoder. The inner and outer bounds on the fundamental trade-off of cache memory size and update rate are provided. For file selection networks with uniform requests, the derived bounds recover the rates established by~\cite{maddah2014fundamental,maddah2014decentralized}. The exact trade-off between the memory and load of delivery is characterized~\cite{yu2018exact} for the uncoded placement. The coded caching problem has also been studied in various setups, including online caching~\cite{pedarsani2015online}, device-to-device caching~\cite{ji2015wireless,ji2015fundamental}, caching with nonuniform demands~\cite{niesen2016coded,zhang2017coded}, coded cache placement~\cite{reisizadeh2018erasure,reisizadeh2019subspace,chen2014fundamental,wei2017novel}.

All of the aforementioned works assume that the delivery phase takes over a perfect channel. However, in practice, we are dealing with noisy broadcast channels, especially for wireless communication systems. For the wireless setup, various types of channel models have been studied, such as cache-aided interference channels~\cite{maddah2019cache,naderializadeh2017fundamental,hachem2016layered}, caching on broadcast channels~\cite{timo2015joint,bidokhti2016erasure}, erasure and fading channels~\cite{bidokhti2017benefits, amiri2018cache, ngo2018scalable}, and channels with delayed feedback with channel state information~\cite{zhang2015coded,zhang2017fundamental}. The cache-aided communications problem is modeled as a joint cache-channel coding problem in~\cite{timo2015joint}. The delivery phase takes place over a memoryless erasure broadcast channel. It is shown that using unequal cache sizes and joint cache-channel coding improves system efficiency when the users experience different channel qualities. The capacity-memory trade-off of the $K$-user broadcast channel is studied when each user is equipped with a cache. It is optimal to assign all the cache memory to the weakest user for the small total cache size. On the other hand, for the large cache size, it is optimal to assign a positive portion of the cache to each user where weaker users have access to a larger cache memory than stronger users. Another wireless communication model is considered in~\cite{zhang2015coded} where a $K$-antenna transmitter communicates to $K$ single receiver antenna. It is shown that the combination of caching with a rate-splitting broadcast approach can reduce the need for channel state information at the transmitter.

Note that in a fast-fading environment sending the channel state information (CSI) from the receiver to the transmitter over a feedback link is difficult. So, it is more reasonable to consider broadcast channels with no CSI. The ergodic capacity region of a $K$-user binary deterministic version of the time-varying memoryless fading broadcast channel ($K$-DTVBC) introduced by~\cite{avestimehr2011wireless} is studied in~\cite{david2012fading,yates2011k}. Depending on the instantaneous channel strength, each user only receives the most significant bits of the transmit signal. Using the insight from the $K$-DTVBC model, an outer bound to the Gaussian fading BC capacity region is derived. 

\noindent\textbf{Contributions.} 
In this work, we study a communication model over the $K$-DTVBC for $n$ channel use where each user is equipped with a cache. The transmitter has some source rate per channel for each file in its database. A fixed fraction of each file is available in each user's cache. Here, we focus on a class of uncoded cache placement schemes. After the completion of this phase, each user demands a file. Then, the transmitter forms broadcasting messages such that each user can decode his desired file. The main challenge for the transmitter is to assign the signal levels to the (broadcasting) messages intended for each user without having access to the realization of the channels, which consists of the number of bits delivered to each user. Note that a fast-fading environment needs coding for reliable communication, where the capacity of the channel is achievable using channel codes with sufficiently large block lengths. Thus, we study the asymptotic behavior of the system, where we allow the size of messages in the pre-fetching and fetching phases will grow with the communication block length. This leads us to deal with a joint source-channel coding problem where the transmitter has a certain source rate per channel use. We characterize the maximum achievable source rate for the two-user and the degraded $K$-user problems.
Then, we provide an upper bound for the source rate. Finally, we discuss an achievable scheme with the linear programming (LP) formulation to show the looseness of the characterization for $K>2$.  

\noindent\textbf{Outline of the Paper.} In the following, we formulate the problem in Section~\ref{sec:problem}, and present the main results in Section~\ref{sec:results}, whose proofs are presented in Section~\ref{sec:proof_main_2}-\ref{sec:proof-degraded-conv}. In Section~\ref{sec:LP}, we provide an achievable source rate through the LP formulation and then show the achieved information-theoretic bound is not tight in general. Finally, we conclude the paper in Section~\ref{sec:conc}.

\noindent\textbf{Notation.} 
Throughout this paper, we denote the set of integers $\{1,2,\ldots,N\}$ by $[N]$ and the set of non-negative real numbers by $\mathbb{R}^{+}$. For a binary vector of length $\nB$, i.e.,  $X\in \mathbb{F}_2^\nB$, and a pair of integers $a<b$, we use the short hand notation $X(a:b)$ to denote $[X(a), X(a+1), \ldots, X(b)]$. We use $(a,b]$ to refer to the interval $(a,b]:=\{x\in \mathbb{R}: a<x\leq b\}$, and its scaled and shifted version is defined as ${\alpha+\beta(a,b]:=(\alpha+\beta a, \alpha+\beta b]}$. For a set of real numbers $\cI$, we use $|\cI|$ to denote its Lebesgue measure, e.g., $|(a,b]|:=b-a$ denotes the length of the interval.  The all-ones and all-zeros vectors are defined as ${\ones_n:=(1,1,\ldots,1)\in \mathbb{R}^{n\times 1}}$ and $\zeros_n:=(0,0,\ldots,0)\in \mathbb{R}^{n\times 1}$, respectively. For a real number $x\in \mathbb{R}$, we denote its floor and ceiling by $\lfloor x \rfloor$ and $\lceil x \rceil$, respectively. The  fractional part of $x$ is denote by $\{x\}:=x-\lfloor  x \rfloor$. Finally, for $n,k\in \mathbb{Z}$, the binomial coefficient is defined as $\binom{n}{k}:=\frac{n!}{k!(n-k)!}$, if $0\leq k\leq n$, and $\binom{n}{k}:=0$, otherwise.
\section{Problem Formulation}\label{sec:problem}
In this section, we first introduce the $K$-DTVBC, which is the core of this work. Then, we discuss the joint source-channel coding problem studied in this paper. 
\begin{figure}[t]
	\centering
	\includegraphics[width=0.4\textwidth]{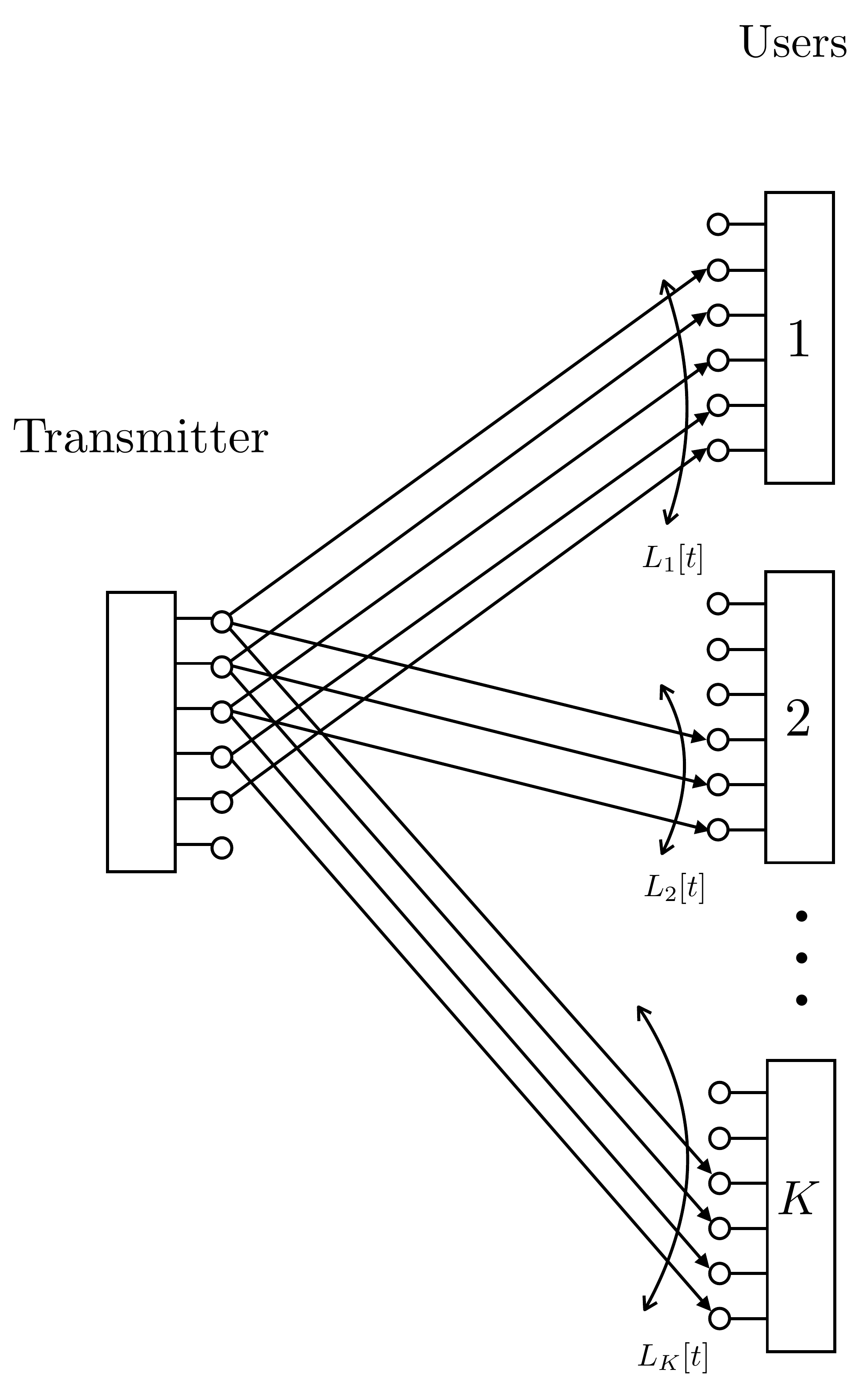}
	\caption{A $K$-user binary deterministic version of the time-varying memoryless fading broadcast channel. The transmitter only knows the statistics, but not the realizations of the generated i.i.d. random sequence $\{L_k[t]\}_{t=1}^{n}$.}
	\label{fig:chnl}
 \vspace{-2mm}
\end{figure} 
\vspace{-2mm}
\subsection{Channel Model}
We are interested in a time-varying broadcast channel, where a transmitter aims at sending one message to each of the $K$ users. We consider the $K$-DTVBC introduced by~\cite{avestimehr2011wireless} as shown in Figure~\ref{fig:chnl}. The  channel is modeled by
\begin{equation}
Y_{k,t} = D^{\nB-L_k[t]}X_{t} = X_t(1:L_k[t]),\qquad k \in [K],
\label{eq:ch-model-detl}
\end{equation} 
where $X_t, Y_{k,t} \in \mathbb{F}_2^\nB$ for $k\in [K]$, and $D$ is a $\nB\times \nB$  shift matrix, given by
\[
D = \begin{bmatrix} 
0 & 0 & 0 & \dots & 0\\
1 & 0 & 0 & \dots & 0\\
0 & 1 & 0 & \dots & 0\\
\vdots & \ddots & \ddots & \ddots & \vdots\\
0 & \dots & 0 & 1 & 0\\
\end{bmatrix}.
\]

\begin{figure}[h]
	\centering
	\includegraphics[width=0.55\textwidth]{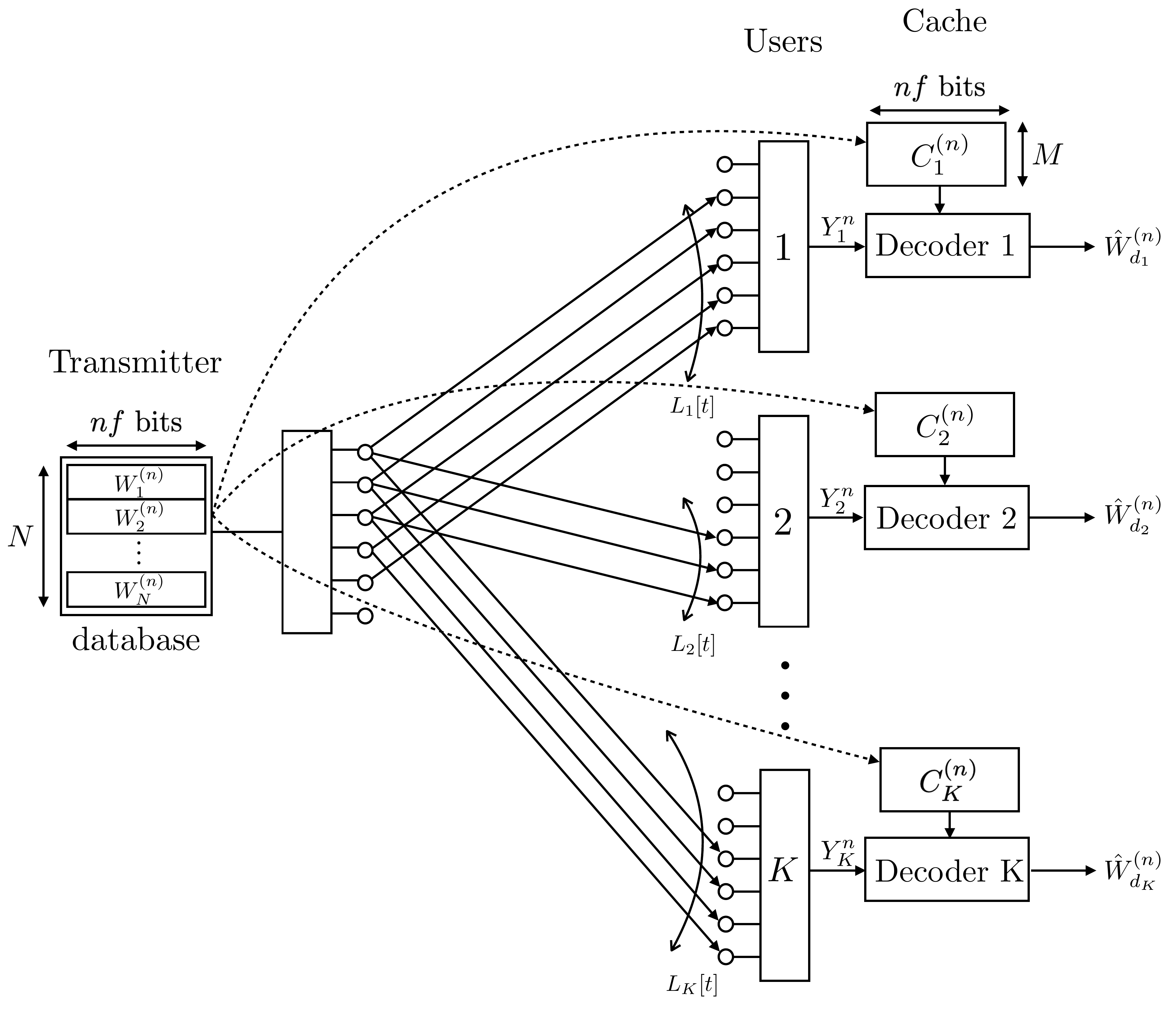}
	\caption{A transmitter containing $N$ files of size $nf$ bits each is connected through a $K$-DTVBC to users each with a cache of size $n M f$ bits. }
	\label{fig:Lu}
 \vspace{-2mm}
\end{figure}
Here $L_k[t]$ with $0\leq L_k[t]\leq \nB$ determines the number of bits delivered to user $k$ at time $t$. The channel state at user $k$, i.e.,  $\{L_k[t]: t = 1,\ldots,n\}$, is an i.i.d. random sequence generated according to some probability mass function (PMF) $P_{L_k}(\ell):= \mathbb{P}[L_k=\ell]$. Intuitively, sending a message $X_t$ of length $B$ bits over a channel with parameters $L_k[t]$, the receiver only receives the $L_k[t]$ most significant bits (MSBs) of $X_t$, and the remaining bits will be erased. This operation can be modeled as the multiplication of the message $X_t$ by $D^{B-L_k[t]}$, where $D$ is the shift matrix. We assume that the channel state information is casually known only to the receivers. However, the transmitter only knows the channel statistics $P_{L_k}(\ell)$, but not the channel realizations. 
\begin{defi}\label{def:ccdf}
We denote the complementary cumulative
distribution function (CCDF)  of $L_k$ for any $\ell \in [B]$ by \begin{align*}
\overline{F}_{L_k}(\ell):= \mathbb{P}[L_k \geq \ell].
\end{align*}
\end{defi}
\noindent For notational simplicity, let
\begin{align*}
        \bof_{L_k} :=
        \begin{bmatrix} 
        \overline{F}_{L_k}(1) \\ \vdots \\ \overline{F}_{L_k}(\nB)
        \end{bmatrix},
\end{align*}
for each user $k \in [K]$.
\begin{defi}\label{def:degraded}
The random variable $L_k$ is \textit{stochastically} larger than $L_v$, if $\ov{F}_{L_k}(\ell) \geq \ov{F}_{L_v}(\ell)$ for every $\ell \in [B]$ and we denote it by $L_k\geq_{\mathsf{st}} L_v $.
\end{defi}
The K-DTVBC channel model for wireless communication simplifies analysis compared to the Gaussian model while still capturing the important features of the problem. This model focuses on signal interactions rather than background noise since networks often operate in interference-limited scenarios. The deterministic model operates on a finite-field, makes it simpler, and provides a complete characterization of network capacity. The insights gained from the deterministic analysis can be applied to find approximately optimal communication schemes for Gaussian relay networks. The analysis of deterministic networks not only guides coding schemes for Gaussian channels but also offers useful proof techniques.
The capacity region of the $K$-DTVBC is derived in~\cite{yates2011k}. In this work, we focus on a cache-aided version of this problem, where the users are equipped with a cache that can pre-fetch part of the messages. In contrast, to~\cite{yates2011k}, where the capacity region is characterized, we are interested in the symmetric rate, as it is standard to consider equal file sizes in file delivery systems. 

In the majority of the existing literature on coded caching, a perfect channel is assumed between the transmitter and the users. Hence, the focus is on minimizing the design of the placement and delivery phases to minimize the load on the perfect channel~\cite{maddah2014fundamental,maddah2014decentralized,ji2015wireless,ji2015fundamental,yu2018exact}. Here, we are dealing with a fast fading channel which requires coding for reliable communication. Hence, we allow for a large code length and study the asymptotic behavior of the channel. Consequently, the size of the message(s) will grow with the communication block length. This leads to a joint source-channel coding problem~\cite{shannon1959coding, gray1974source, gastpar2003code}. More precisely, we consider a communication scenario over $n$ \emph{channel uses}, where the transmitter has a library of $N$ files, each of size $nf$ bits. Each user is equipped with a cache that can pre-fetch up to $n M f$ bits (before the actual request of the user is revealed), and the goal is to send one requested file to each user reliably. We are interested in characterizing the maximum source rate $f$ for which, and for sufficiently large block length $n$, a reliable communication scheme can be devised. A similar joint source-channel coding approach is used to study the original coded caching problem with common rate and side information in~\cite{lim2017information}. Further details of the cache model are discussed in the next section. 

\subsection{Joint Source-channel Coding Framework}
Let us consider a communication scenario over the \mbox{$K$-DTVBC} for $n$ channel uses. The transmitter has some source rate $f\in \mathbb{R}^+$ per channel use that generates $N$ files, namely, $W^{(n)}_i$ for $i\in[N]$. This means the transmitter has access to a database of $N$ mutually independent files $W^{(n)}_1,\ldots, W^{(n)}_N$ each of size $nf$ bits, i.e.,
\begin{equation*}
W^{(n)}_i\in \{1,2,\ldots,2^{nf}\}, \quad i\in [N].
\end{equation*} 
We assume each user $k$ is equipped with a cache, which can pre-fetch part of the files. The size of the content is proportional to the communication block length. More precisely, we assume that user $k$ has a cache $C^{(n)}_k$ of size $nMf$ bits, for $k\in[K]$. In the placement phase, the cache memory of each user is filled with \emph{uncoded} bits of the files; that is, the content of the cache $C^{(n)}_k$ can be partitioned into raw (uncoded) bits of the files. 
\begin{defi}
    A \emph{caching strategy} $\cache$ for a normalized cache size $\mu=M/N$ and a network with $K$ users consists of $K$ collections of intervals in $(0,1]$. More precisely, ${\cache=(\bc_1,\bc_2,\dots, \bc_K)}$ where 
\begin{itemize}
    \item $\bc_k = \bigcup_{\ell \in [\nn_k]} \cI_{k,\ell}$,
    \item $\nn_k$ is a finite positive integer number for every $k\in [K]$,
    \item $\cI_{k,\ell}=(a_{k,\ell}, b_{k,\ell}] \subseteq (0,1]$ where $b_{k,\ell}\leq a_{k,\ell+1}$ for every $\ell\in [K\!-\!1]$, and
    \item $\sum_{\ell\in[\nn_k]} |\cI_{k,\ell}| =\mu$, for every $k\in [K]$. 
\end{itemize}
For a file $W=(W(1),W(2), \dots, W(F))\in \mathbb{F}_2^F$ of length $F$ bits, we define 
\[W(\bc_k) \!:=\! \bigcup_{\ell\in [\nn_k]} \{W(\lceil a_{k,\ell}F\rceil+1), \cdots,  W(\lfloor b_{k,\ell}F\rfloor) \}.\]
\end{defi}
For a given source rate $f$, block length $n$, family of files $\left\{W_i^{(n)}\right\}$, and caching placement strategy $\cache$, the cache content of user $k\in [K]$ is given by 
\begin{align}\label{eq:cache}
C^{(n)}_k \!\! &:=\!\bigcup_{i\in [N]} C^{(n)}_{k,i} =\!\bigcup_{i\in [N]} W^{(n)}_i(\bc_k) \\
& =\! \!\bigcup_{i\in [N]} \hspace{-2pt}\bigcup_{\ \ell \in [\nn_k]} \!\!\!\hspace{-2pt}\left\{\hspace{-2pt}W^{(n)}_i(\lceil nf a_{k,\ell}\rceil \!+\!1),\dots, W^{(n)}_i(\lfloor nf b_{k,\ell}\rfloor) \hspace{-2pt}\right\}.\nonumber
\end{align}
This implies that 
\begin{align}
H\left(C^{(n)}_{k,i}\right) & \leq  \sum_{\ell\in[\nn_k]} \left(\lfloor nf b_{k,\ell} \rfloor - \lceil nf a_{k,\ell}\rceil \right)\nonumber\\
& \leq \sum_{\ell\in[\nn_k]} nf |\cI_{k,\ell}| =  \mu nf.
\end{align}
Therefore, we get
\begin{align*}
H\left(C_k^{(n)}\right) &= H\left( C^{(n)}_{k,1}, C^{(n)}_{k,2}, \dots, C^{(n)}_{k,N}\right)\\
&\leq H\left( C^{(n)}_{k,1}\right) +H\left(C^{(n)}_{k,2}\right)+ \cdots +  H\left(C^{(n)}_{k,N}\right)\\
&\leq \sum_{i=1}^N n \mu f = nM f. 
\end{align*}
Moreover, from the definition of the cache content in~\eqref{eq:cache} and the independence of files, we can write
\begin{align*}
    & H\lp C^{(n)}_{k,i}\middle| W^{(n)}_i \rp=0,\\
    & I\lp C^{(n)}_{k,j};W^{(n)}_i\rp=0, \quad j \neq i.
\end{align*}
We define $\bc_\cS:= \bigcup_{u\in \cS} \bc_u = \bigcup_{u\in \cS} \bigcup_{\ell\in [\nn_u]} \cI_{u,\ell}$ for every $\cS\subseteq [K]$ and the \emph{caching tuple} $\bmu :=(\mu_{\cS}: \cS\subseteq [K])$ where $\mu_{\cS}:= |\bc_\cS|$. We also use $C^{(n)}_{\cS,i}$ to refer to the collection of all the parts of file $i$ cached by the users in the subset $\cS\subseteq [K]$, i.e.,  $C^{(n)}_{\cS,i} = \bigcup_{u\in \cS} C^{(n)}_{u,i}$. Therefore, we have 
\begin{align}\label{eq:H-CS}
H\left(C^{(n)}_{\cS,i}\right) \leq \mu_\cS nf,
\end{align}
for every $i\in [N]$.
After the completion of the placement phase, each user requests one of the $N$ files, where all files are equally likely to be requested. We denote $d_k\in [N]$ as the index of the file requested by user $k\in[K]$ and the sequence of all requests by $\boldsymbol{d}=(d_1,\ldots,d_K)$. 
Once the requests are revealed to the transmitter, it forms a broadcasting message ${X^{n}=(X_1,X_2,\dots, X_n)=\psi_{\boldsymbol{d}}^{(n)} \left(W^{(n)}_1,\ldots,W^{(n)}_N; C^{(n)}_{[K]}\right)}$, where
\begin{equation*}
\psi_{\boldsymbol{d}}^{(n)}\!\!:\!\! \{1,2,\ldots,2^{nf}\} ^N \!\times\! \{1,2,\dots, 2^{nMf}\}^K \! \rightarrow \{1,2,\ldots,2^{\nB}\}^n\!,
\end{equation*} 
and transmits $X_t$ over the broadcast channel during the $t$th channel use  of the delivery phase, for $t=1,\dots, n$. Upon receiving  $Y_k^n=(Y_{k,1},Y_{k,2}, \ldots, Y_{k,n})$, user $k\in [K]$ should be able to decode its desired file using its cache content $C^{(n)}_k$ and the received message $Y^n_k$ (see Figure~\ref{fig:Lu}), i.e., 
\begin{equation*}
\hat{W}^{(n)}_{d_k} = \phi_k^{(n)}\left(Y^n_k,C^{(n)}_k\right).
\end{equation*}
Here, we define the overall decoding error probability as ${P_e^{(n)}\!:=\! \sum_{k=1}^{K}\! \mathbb{P}\left[\hat{W}^{(n)}_{d_k}\!\neq\! W^{(n)}_{d_k}\right].}$

\begin{defi}
    For a given caching strategy $\cache$ and a (distinct) request profile $\bd$, a source rate $f(\cache,\bd)$ is called achievable if there exists a sequence of encoding and decoding functions ${\left\{\left(\psi^{(n)}, \phi_1^{(n)}, \ldots,\phi_K^{(n)}\right)\right\}_n}$, for which $P_e^{(n)} \rightarrow 0$ as $n$ grows. 
\end{defi}

Here, our goal is to characterize the maximum achievable source rate $f(\cache,\bd)$ for a given $K$-DTVBC with channel statistics, $\bof_{L_k}$ for $k\in [K]$. Note that the cache placement is fixed prior to the users' demands, and we are not designing the cache contents of users based on the requested files. 

For every subset of users $\cS\subseteq [K]$ and file index $i\in[N]$, we define ${W_{i,\cS}^{(n)} = \bigcap_{k\in \cS} W_i^{(n)} (\bc_k) = W_i^{(n)} \left(\bigcap_{k\in \cS}  \bc_k\right)}$, to be the sections of file $W_i^{(n)}$ which are cached at all users in $\cS$.

Next, inspired by the central cache placement strategy of~\cite{maddah2014fundamental}, we introduce the \emph{central caching strategy} $\cent$. For a subset $\cS\subseteq [K]$ with $|\cS|=s$,  let $\n(\cS) \in \left\{1,2,\dots, \binom{K}{s}\right\}$ be the rank of $\cS$ among all subsets of $[K]$ of size $s$, according to the \emph{lexicographical order}.

\begin{defi}\label{def:cent}
    For every  $\cS\subseteq [K]$ with $|\cS| = s$, define
    \begin{align*}
        \mathcal{J}_{\cS} := \left( \frac{\n(\cS)}{\binom{K}{s}}, \frac{\n(\cS)+1}{\binom{K}{s}}\right].
    \end{align*}
    Then, for a network with $K$ users and normalized cache size $\mu\in[0,1]$, we define the central caching strategy ${\cent:=(\ccent{1},\dots, \ccent{K})}$ where 
\begin{align*}
    \ccent{k} \!:=\! \left(\!\bigcup_{\substack{\cS\subseteq [K] \\ |\cS|= \lfloor \mu K\rfloor  \\ \cS\ni k }} \!\!\!\!(1-\lambda)\mathcal{J}_{\cS}\right) \!\cup\! \left(\bigcup_{\substack{\cT\subseteq [K] \\ |\cT|= \lfloor \mu K\rfloor +1 \\ \cT\ni k }} \!\!\!\!\!((1-\lambda)\!+\!\lambda\mathcal{J}_{\cT}) \!\right)\!, 
\end{align*}
and\footnote{Note that  $\mu K = (1-\lambda)   \lfloor \mu K\rfloor + \lambda \lceil \mu K\rceil$} $\lambda = \{\mu K\}$.\label{def:central}
\end{defi}

Note that for any set of users $\cQ\subseteq[K]$,  we have 
\begin{align}\label{eq:cent-union} \mu_{\cQ}^{\mathsf{cent}} & = \left|\ccent{\cQ}\right|\nonumber \\
& = \left|\bigcup_{k\in \cQ} \ccent{k}\right| \nonumber\\
& =(1-\lambda) \left(1- \frac{1}{\binom{K}{\lfloor \mu K \rfloor}} \sum_{\substack{\cS\subseteq [K]\\|\cS|=\lfloor \mu K \rfloor \\ \cS \cap \cQ=\emptyset } } 1 \right)+   
  \lambda \left(1- \frac{1}{\binom{K}{\lfloor \mu K \rfloor+1}} \sum_{\substack{\cT\subseteq [K]\\|\cT|=\lfloor \mu K \rfloor+1 \\ \cT \cap \cQ=\emptyset } } 1 \right) \nonumber\\
  &=(1\!-\!\lambda) \left(1\!-\! \frac{\binom{K-|\cQ|}{\lfloor \mu K \rfloor}}{\binom{K}{\lfloor \mu K \rfloor}} \right)  \!+\!  \lambda  \left(1\!-\! \frac{\binom{K-|\cQ|}{\lfloor \mu K \rfloor+1}}{\binom{K}{\lfloor \mu K \rfloor+1}} \right).
\end{align}
Moreover, we have 
\begin{align}\label{eq:cent-intersection}
\left|\bigcap_{k\in \cQ} \ccent{k}\right|  &=\!(1\!-\!\lambda)   \frac{1}{\binom{K}{\lfloor \mu K \rfloor}} \!\!\!\sum_{\substack{\cS\subseteq [K]\\|\cS|\!=\lfloor \mu K \rfloor \\ \cQ  \subseteq \cS} } \! \! \! 1   +  \! 
  \lambda   \frac{1}{\binom{K}{\lfloor \mu K \rfloor+1}} \!\!\!\! \sum_{\substack{\cT\subseteq [K]\\|\cT|=\lfloor \mu K \rfloor+1 \\ \cQ  \subseteq \cT } } \!\!\!\! 1  \nonumber\\
  &=(1-\lambda)   \frac{\binom{K-|\cQ|}{\lfloor \mu K \rfloor-|\cQ|} }{\binom{K}{\lfloor \mu K \rfloor}}  +  \lambda   \frac{\binom{K-|\cQ|}{\lfloor \mu K \rfloor+1-|\cQ|}}{\binom{K}{\lfloor \mu K \rfloor+1}} .
\end{align}

\section{Main Results}\label{sec:results} 
In this section, we present the main results of this paper, organized according to the level of generalization of the setting.

We first characterize the maximum achievable source rate for the $2$-user DTVBC. 
\begin{theorem}[Two-User (Non-Degraded) BC]\label{thm:2-User}
    For a $2$-DTVBC with a caching strategy $\cache$, a distinct request profile $\bd$, and  $\mu\leq \frac{1}{2}$, any achievable source rate is upper bounded by
\begin{align}\label{eq:f-2User}
     \fs \!= \!\min &\left\{ \min_{\w \geq 1}\frac{\w R_1(\w) +  R_2(\w)}{\w (1\!-\!\mu)+(1\!-\!2\mu)}, \min_{0\leq \w \leq 1}\frac{R_1(\w) \!+\! \frac{1}{\w} R_2(\w)}{(1\!-\!2\mu) \!+\! \frac{1}{\w} (1\!-\!\mu)}  \right\}. 
\end{align}
Moreover, if $\frac{1}{2}\leq \mu\leq 1$, any achievable source rate is upper bounded by 
\begin{align}\label{eq:f-2User-2}
    \fs = \min \lc \frac{\sum_{\ell=1}^{B} \ov{F}_{L_1}(\ell)}{1-\mu}, \frac{\sum_{\ell=1}^{B} \ov{F}_{L_2}(\ell)}{1-\mu}\rc,
\end{align}
where 
\begin{align}\label{eq:R1-R2}
\begin{split}
        R_1(\w)&:=\!\!\! \sum_{\ell\in \cL_1(\w)} \overline{F}_{L_{1}}(\ell)\\
        R_2(\w) &:=\!\!\! \sum_{\ell\in \cL_2(\w) } \overline{F}_{L_{2}}(\ell),   
        \end{split}
    \end{align}
and summations are over $\cL_1 (\w) :=  \{\ell: \w \ov{F}_{L_1}(\ell) \geq \ov{F}_{L_2}(\ell)\}$, and ${\cL_2 (\w) := \{\ell: \w \ov{F}_{L_1}(\ell) < \ov{F}_{L_2}(\ell)\}}$. Moreover, the source rates in~\eqref{eq:f-2User} and~\eqref{eq:f-2User-2} are achievable for the central caching strategy $\cent$.
\end{theorem}
The proof of Theorem~\ref{thm:2-User} is provided in Section~\ref{sec:proof_main_2}.

In the upper bound presented in Theorem~\ref{thm:2-User}, we intuitively enhance both the cache and channel strengths for each user and compare the two possible settings. 

Now, let us consider a more general setting where a network is serving $K$ users. In a $K$ user setting, we can characterize the maximum achievable source rate when the channels from the transmitter to the users are degraded (see Definition~\ref{def:degraded}). In the following theorem, we provide an LP optimization problem for the maximum achievable source rate of the degraded $K$-DTVBC with the caching strategy $\cent$.
\begin{theorem}[$K$-User Degraded BC]\label{thm:degraded}
     For a degraded $K$-DTVBC ${L_K\geq_{\mathsf{st}}\cdots \geq_{\mathsf{st}} L_1}$ and a normalized cache sizes $\mu$ satisfying $K\mu \in \mathbb{N}$, 
     with the central caching strategy $\cent$ and a distinct request profile $\bd$, the maximum achievable source rate
     is given by
    \begin{align}\label{eq:fwKUser_degraded}
    &\max_{\{z_{\ell,k}\}} \ \bar{f},\\
      \begin{split}\label{eq:const_thm2}
     &\textrm{s.t.}\quad  \lp 1-\tcents_{[k]}\rp \bar{f} \leq  \sum_{\ell=1}^{B} z_{\ell,k}\ov{F}_{L_k}(\ell), \quad \forall k\in [K],\\
     &\phantom{\textrm{s.t.}\quad} z_{\ell,k}  \geq  0, \quad \forall k\in [K], \forall \ell\in [B],\\
     &\phantom{\textrm{s.t.}\quad} \sum_{k=1}^{K} z_{\ell,k} \leq  1, \quad \forall\ell\in [B].
      \end{split}
\end{align}
\end{theorem}
The proof of Theorem~\ref{thm:degraded} is presented in two parts. The proof of achievability part is presented in Section~\ref{sec:proof-degraded-ach}, and its converse proof is provided in Section~\ref{sec:proof-degraded-conv}. 
We note that the achievability proof of Theorem~\ref{thm:degraded} is based on the LP-based method which is discussed in Section~\ref{sec:LP}.
Next, we will now present an illustrative example of the degraded $K$-DTVBC with three users ($k=3$). This example will be helpful in establishing the notation and following the proof techniques. 

\addtocounter{equation}{3}
\begin{figure*}
\begin{align}
      &\bA_{\pi}\!=\! \begin{bmatrix}
      \begin{array}{c c  c c c|c}
      \bof_{L_{\pi(1)}} & \zeros_B & \cdots & \zeros_B & \zeros_B & \lp \mu_{\pi(1)}-1\rp \!\bfI\\
      \zeros_B & \bof_{L_{\pi(2)}} & \cdots & \zeros_B & \zeros_B & \lp \mu_{\pi([2])}-1\rp\!\bfI\\
      \vdots & \vdots & \vdots & \vdots & \vdots &  \vdots\\
      \zeros_B & \zeros_B & \cdots & \zeros_B & \bof_{L_{\pi(K)}} & \lp \mu_{\pi([K])}-1\rp\!\bfI\\
      \hline
       \mu_{\pi([2])}\!-\!1 &  1\!-\!\mu_{\pi(1)} & \cdots & 0 & 0 & 0\\
      \vdots & \vdots & \vdots & \vdots & \vdots & \vdots \\
      0 & 0 & \cdots & \mu_{\pi([K])}\!-\!1 & 1\!-\!\mu_{\pi([K\!-\!1])} & 0
    \end{array}
    \end{bmatrix}\!,\label{eq:Matrix-A}
 \end{align}
 \hrule
 \end{figure*}
 \addtocounter{equation}{-4}
 
\begin{example}[\textit{Degraded Channel Case}]\label{ex:deg_LP-mu=1/3}
We consider a network with $K=3$ users $N=3$ files, namely $\{W_1,W_2,W_3\}$, and $\nB=3$ signal levels. The channel statistics of the three users are given by the cumulative  distribution  functions as
	\begin{align*}
	&\overline{F}_{L_1} =[0.5,0.4,0.3]^T, \cr
	&\overline{F}_{L_2} = [0.7,0.5,0.4]^T,\cr
	&\overline{F}_{L_3} = [0.9,0.6,0.5]^T.
	\end{align*}
    That is, the first user receives the top level with probability $0.5$, but the bits sent over all three $\nB=3$ levels are delivered at this same user with probability $0.3$. 
    Considering the caching strategy $\cent$ with a caching factor of $\mu = 1/3$, it can be defined as follows
    \begin{align*}
        \bc_1 = (0,1/3],\quad \bc_2 = (1/3, 2/3], \quad \bc_3 = (2/3,1].   
    \end{align*}
    Hence, the placement strategy $\cent$ implies that the cached parts of the files at different users are \textit{disjoint}, i.e., ${W^{(n)}_{i} = \bigcup_{k=1}^{3} W^{(n)}_{i,k}}$ for every $i\in [N]$ where $W^{(n)}_{i,k}$ the part of file $W^{(n)}_i$ cached \textit{exactly} by user $k$.  
	
    Without loss of generality, assume user $k$ is interested in file $W^{(n)}_k$, for $k\in\{1,2,3\}$. Here, we have ${\tcents_{\{1\}}=1/3}$, ${\tcents_{\{1,2\}}=2/3}$, and $\tcents_{\{1,2,3\}}=1$. 
    The coefficients ${[\bz]_{(\ell,k)} := z_{\ell,k}}$ that provide the optimum solution of~\eqref{eq:fwKUser_degraded} are give by 
    \begin{align*}
    &\hspace{5mm}\begin{tabular}{ p{5mm}p{5mm}p{4mm}p{4mm}p{4mm}p{4mm}p{4mm}p{4mm}p{4mm}} 
    	 {$\scriptstyle(1,1)$} &  {$\scriptstyle(1,2)$} & {$\scriptstyle(1,3)$} &  {$\scriptstyle(2,1)$} &  {$\scriptstyle(2,2)$} &  {$\scriptstyle(2,3)$} &  {$\scriptstyle(3,1)$} &  {$\scriptstyle(3,2)$} &  {$\scriptstyle(3,3)$} \end{tabular}\nonumber\\
	    &\bz \!=\! \left[ 
     \begin{tabular}{p{5mm}p{5mm}p{4mm}p{4mm}p{4mm}p{4mm}p{4mm}p{4mm}p{4mm}}
	    $\!\!0.37$ & $0.63$ & $\ 0$ & $\ 1$ & $\ 0$ & $\ 0$ & $\ 1$& $\ 0$ & $0\!\!$  
     \end{tabular}
	   \!\!\! \right]^T\!\!,
	\end{align*}
    with the optimum source rate
    \begin{align*}
        \fs \!=\! \frac{\sum_{\ell=1}^{3} z_{\ell,k}\ov{F}_{L_1}(\ell)}{1-\tcents_{\{1\}}} \!=\! \frac{\sum_{\ell=1}^{3} z_{\ell,k}\ov{F}_{L_2}(\ell)}{1-\tcents_{\{1,2\}}} \!=\! 1.326.
    \end{align*}
     \hfill $\diamond$
\end{example}
Finally, we can present our result for the most general case, which is the (non-degraded) $K$-DTVBC. The following theorem provides an upper bound for the source rate of the $K$-DTVBC with \textit{any} given caching strategy $\cache$.
\begin{theorem}[$K$-User (Non-Degraded) BC]\label{thm:Kuser-deg-BC}
    Any achievable source rate of the $K$-DTVBC for a given cache placement strategy $\cache$ and a (distinct) request profile $\bd$ is upper-bounded by 
\begin{align}\label{eq:fwKUser}
    \! f(\cache,\bd) \!\leq\! \fs(\cache,\bd) \!:=\! \min_{\bfw \geq 0} \frac{\sum_{\ell=1}^\nB \max_{k}  \w_{\pi(k)} \ov{F}_{L_{\pi(k)}} (\ell)}{\sum_{k=1}^{K}
	\omega_{\pi(k)}\lp 1-\mu_{\pi([k])} \rp }, 
\end{align}
	where $\bfw := (\w_1,\ldots,\w_K) \in [0,\infty)^K$ is a non-negative vector of length $K$, ${\pi([k]):=\{\pi(1),\ldots, \pi(k)\}}$, and $\pi:[K]\rightarrow [K]$ is the permutation that sorts $\bfw$ in the non-increasing order.
\end{theorem}
The proof of Theorem~\ref{thm:Kuser-deg-BC} is presented in Section~\ref{sec:proof_main_1}.

Note that the upper bound in the theorem is given as a min-max problem, for which the evaluation of the optimum point can be computationally challenging. In the following proposition, we show that the upper bound in~\eqref{eq:fwKUser} can be evaluated by solving $K!$ linear programming problems, each corresponding to a permutation of users. Hence, denoting $t_{\mathsf{LP}}$ as the run time for each LP, the complexity of evaluation of the upper-bound in~\eqref{eq:fwKUser} is $K!\times t_{\mathsf{LP}}$. 
\begin{proposition}\label{pro:LP-fbnd}
The min-max problem in~\eqref{eq:fwKUser} is equivalent to 
\begin{align}\label{eq:LP_form}
    \fs(\cache,\bd) = \min_{\pi\in \Pi} \fs_{\pi}(\cache,\bd),
\end{align}
where 
\begin{align}\label{eq:LP_form_2}
    \fs_{\pi}(\cache,\bd) := & \min_{\bx \in \mathbb{R}^{K\!+\!B}} \ [\zeros_{K}^T,\ones_{B }^T]\bx,\\
    & \ \ \ \ \textrm{s.t.}\quad \bA_{\pi} \bx \leq \zeros,\nonumber\\
    & \ \ \ \ \ \ \phantom{\textrm{s.t.}\quad} \bb \bx = 1,\nonumber\\ 
    & \ \ \ \ \phantom{\textrm{s.t.}\quad} -\bx \leq \zeros\nonumber, 
\end{align}
and $\Pi$ is the set of all permutations over $[K]$. The matrix $\bA_{\pi}\in \mathbb{R}^{(KB + K\!-\!1)\times(K+B)}$ is given in~\eqref{eq:Matrix-A} at the top of this page. Moreover,  
\addtocounter{equation}{1}
\begin{align}
\bb=[\ones_K ^T, \zeros_B^T]\in \mathbb{R}^{1\times (K+B)},
    \label{eq:vector-b}
\end{align}
and $\bfI \in \mathbb{R}^{B\times B}$ is the identity matrix. Note that for each permutation $\pi\in\Pi$, the LP problem in~\eqref{eq:LP_form_2} consists of $K\!+\!B$ variables and $K(B\!+\!2)\!+\!B$ constraints. 
\end{proposition}
The proof of Proposition~\ref{pro:LP-fbnd} is provided in Section~\ref{sec:proof_main_1}. 

\section{An Achievable Scheme: LP Formulation} \label{sec:LP}
In this section, we provide an achievable scheme, which is based on Linear programming. This scheme is optimum for the degraded broadcast channels (as claimed in Theorem~\ref{thm:Kuser-deg-BC}). However, in an illustrative example, we show that there is a gap between the achievable rate of the proposed scheme and the upper bound in~\eqref{eq:fwKUser}. This implies that either the achievable scheme is not optimum, or the upper bound is not tight. Consequently, the exact characterization of the optimum source rate of a non-degraded $K$-DTVBC with $K>2$ remains as an open problem for future works. 

Similar to~\cite{maddah2014fundamental}, we focus on specific normalized cache sizes $\mu$ such that  $t:= K\mu =KM/N \in \mathbb{N}$. Let us assume that each user $k\in[K]$ requests file $W^{(n)}_k$. The delivery scheme of~\cite{maddah2014fundamental}  consists of broadcasting coded packets to serve multiple users simultaneously. Each coded packet is intended for a group of users $\cS\subseteq \{1,\dots, K\}$ with $|\cS|=t+1$. More precisely, we have
\begin{align*}
    W^{(n)}_\mathcal{\cS} = \bigoplus_{k\in \cS} W^{(n)}_{d_k,\cS\setminus\{k\}}.
\end{align*}
We aim at sending each coded packet $W^{(n)}_\mathcal{\cS}$ to all users $k\in \cS$. To this end, we devise a bit allocation strategy ${\by\!=\!\{y_{\ell,\cS}\!:\! \ell\in [\nB], \cS \subseteq [K], |\cS|\!=\!t\!+\!1\}}$, where {$0 \leq y_{\ell,\cS} \leq 1$} is 
a variable indicating the fraction of time that level $\ell$ of the channel is used to transmit coded message $W^{(n)}_{\cS}$. Note that a signal level $\ell$ can be shared between  \textit{multiple} coded message  $W^{(n)}_{\cS}$. For a feasible allocation policy $\by$, in each level, the sum of time {fractions} allocated to all coded messages should not exceed $1$, that is, 
\begin{align}
  \sum_{\substack{\cS\subseteq [K]\\|\cS|=t+1 }}  y_{\ell,\cS} \leq 1, \qquad \forall \ell\in [\nB].
    \label{eq:LP-level}
\end{align}
To ensure successful decoding of the sub-message $W^{(n)}_{\cS}$ by every user in $k\in\cS$, it is necessary to assign sufficiently large values to the coefficients $y_{\ell,\cS}$. Recall that for a given common source rate
$f$, the rate of $W^{(n)}_\cS$ is given by 
$f/\binom{K}{t}$. 
Then, a coded message $W^{(n)}_{\cS}$ is decodable at user $k\in \cS$ if
\begin{align}
    \sum_{\ell=1}^\nB \overline{F}_{L_k}(\ell) y_{\ell,\cS}  \geq \frac{f}{\binom{K}{t}}.
    \label{eq:LP-decode}
\end{align}
\addtocounter{equation}{3}
\begin{figure*}
\begin{align}\label{eq:mat_G}
\bG = 
\begin{array}{cc} &
\begin{array}{cccccccccccc} \hspace{-15pt}& \!\!
	     {\scriptstyle(1,\{1,2\})} &\!\! {\scriptstyle(1,\{1,3\})} &\!\!\!\! {\scriptstyle(1,\{2,3\})} &\!\! {\scriptstyle(2,\{1,2\})} &\!\!\!\! {\scriptstyle(2,\{1,3\})} &\!\! {\scriptstyle(2,\{2,3\})} &\!\! {\scriptstyle(3,\{1,2\})} &\!\! {\scriptstyle(3,\{1,3\})} &\!\! {\scriptstyle(3,\{2,3\})} & \!\!
	     m \\
\end{array}
\\
\begin{array}{cccccc}
{\scriptstyle(1,\{2\})} \\ {\scriptstyle(2,\{1\})} \\ {\scriptstyle(1,\{3\})} \\ {\scriptstyle(3,\{1\})} \\ {\scriptstyle(2,\{3\})} \\ {\scriptstyle(3,\{2\})} 
\end{array}
&
\hspace{-5pt}\left[\hspace{-10pt}
\begin{array}{cccccccccccc}
&\hspace{7pt} -0.9	&\hspace{7pt} 0	 &\hspace{7pt} 0	&\hspace{7pt}-0.3	&\hspace{7pt}0	&\hspace{7pt}0	&\hspace{7pt}-0.3	&\hspace{7pt} 0	&\hspace{7pt} 0	&\hspace{7pt} 0.33\\
&\hspace{7pt} -0.7 &\hspace{7pt}	0&\hspace{7pt}	0&\hspace{7pt}	-0.4&\hspace{7pt}	0	&\hspace{7pt} 0	&\hspace{7pt}-0.4	&\hspace{7pt}0	&\hspace{7pt}0	&\hspace{7pt} 0.33\\
&\hspace{7pt} 0&\hspace{7pt}	-0.9&\hspace{7pt}	0&\hspace{7pt}  0&\hspace{7pt}	-0.3&\hspace{7pt}	0&\hspace{7pt}	0&\hspace{7pt}	-0.3&\hspace{7pt}	0&\hspace{7pt}	0.33\\
&\hspace{7pt} 0 &\hspace{7pt}	-0.5&\hspace{7pt}	0&\hspace{7pt}	0&\hspace{7pt}	-0.5&\hspace{7pt}	0&\hspace{7pt}	0&\hspace{7pt}	-0.5&\hspace{7pt}	0&\hspace{7pt}	0.33\\
&\hspace{7pt} 0 &\hspace{7pt}	0&\hspace{7pt}	-0.7&\hspace{7pt}	0&\hspace{7pt}	0&\hspace{7pt}	-0.4&\hspace{7pt}	0&\hspace{7pt}	0&\hspace{7pt}	-0.4&\hspace{7pt}	0.33\\
&\hspace{7pt} 0 &\hspace{7pt}	0&\hspace{7pt}	-0.5&\hspace{7pt}	0&\hspace{7pt}	0&\hspace{7pt}	-0.5&\hspace{7pt}	0&\hspace{7pt}	0&\hspace{7pt}	-0.5&\hspace{7pt}	0.33
\end{array}
\right],
\end{array}
\end{align}
\begin{align}\label{eq:mat_H}
\bH = 
\begin{array}{cc} &
\begin{array}{cccccccccccc} \hspace{-22pt}& \!\!
	     {\scriptstyle(1,\{1,2\})} &\!\! {\scriptstyle(1,\{1,3\})} &\!\!\!\! {\scriptstyle(1,\{2,3\})} &\!\! {\scriptstyle(2,\{1,2\})} &\!\!\!\! {\scriptstyle(2,\{1,3\})} &\!\! {\scriptstyle(2,\{2,3\})} &\!\! {\scriptstyle(3,\{1,2\})} &\!\! {\scriptstyle(3,\{1,3\})} &\!\! {\scriptstyle(3,\{2,3\})} & \!\!
	     m \\
\end{array}
\\
\begin{array}{cccccc}
{1} \\ {2} \\ {3}  
\end{array}
&
\hspace{-10pt}\left[\hspace{-22pt}
\begin{array}{cccccccccccc}
&\hspace{22pt} 1	&\hspace{22pt} 1	 &\hspace{22pt} 1	&\hspace{22pt}0	&\hspace{22pt}0	&\hspace{22pt}0	&\hspace{22pt}0	&\hspace{22pt} 0	&\hspace{22pt} 0	&\hspace{18pt} 0\\
&\hspace{22pt} 0 &\hspace{22pt}	0&\hspace{22pt}	0&\hspace{22pt}	1&\hspace{22pt}	1	&\hspace{22pt} 1	&\hspace{22pt}0	&\hspace{22pt}0	&\hspace{22pt}0	&\hspace{18pt} 0\\
&\hspace{22pt} 0&\hspace{22pt}	0&\hspace{22pt}	0&\hspace{22pt}  0&\hspace{22pt}	0&\hspace{22pt}	0&\hspace{22pt}	1&\hspace{22pt}	1&\hspace{22pt}	1&\hspace{18pt}	0
\end{array}
\right].
\end{array}
\end{align}
\begin{align}
	    \begin{array}{cccccccccccc}
	     & \!
	     {\scriptstyle(1,\{1,2\})} &\! {\scriptstyle(1,\{1,3\})} &\! {\scriptstyle(1,\{2,3\})} &\! {\scriptstyle(2,\{1,2\})} &\! {\scriptstyle(2,\{1,3\})} &\! {\scriptstyle(2,\{2,3\})} &\! {\scriptstyle(3,\{1,2\})} &\! {\scriptstyle(3,\{1,3\})} &\! {\scriptstyle(3,\{2,3\})} & \!
	     m \\
	    \by = \Big[ &\!
	    \frac{2}{3} &  
	    \frac{ 1}{3} &
	    0 &
	    \frac{1}{12} &
	    0 & 
	    \frac{11}{12} &
	    0 & 
	    \frac{2}{3} &
	    \frac{1}{3} &
	    \frac{3}{2} &
	    \Big]^T\!\!.
	    \end{array}
	    \label{eq:vect_y} 
	\end{align}
\hrule
\end{figure*}
\addtocounter{equation}{-6}

So, we have an optimization problem given by 
\begin{align}
\begin{split}
    \fl := & \max \  f\\
    & \textrm{s.t.} \  \eqref{eq:LP-level}-\eqref{eq:LP-decode}\\
    &\ y_{\ell,\cS} \geq 0, \quad \forall \ell\in [\nB], \forall \cS\subseteq [K], |\cS|=t+1.
\end{split}
\label{eq:LP}
\end{align}
We can write the optimization problem as a linear program. To this end, we define a vector $\by$ of length ${\ml = \nB \binom{K}{t+1}\!+\!1}$, forming by stacking all variables in  ${\{y_{\ell,\cS}\!:\! \ell\!\in\! [\nB], \cS\subseteq [K], |\cS| = t\!+\!1\}}$ along with $f$ at the very last position. The first $\nB\binom{K}{t+1}$ entries of $\by$ are labeled by pairs $(\ell,\cS)$ and we set $\by_m = -f$ as the last entry of $\by$.

To write the constraint 
in~\eqref{eq:LP-level} in matrix form, we can define a matrix $\bH\in \mathbb{R}^{\nB \times m}$, where its rows indexed by $\ell\in [B]$, and its columns are labeled similar to $\by$. Moreover, we have 
\begin{align*}
    \bH_{\b,(\ell,\cS)} &=  
    \begin{cases}
    1 & \textrm{if $\b=\ell$}\\
    0 & \textrm{if $\b\neq\ell$},
    \end{cases}\\
    \bH_{\b,m} &= 0, \qquad \forall \b\in [\nB].
\end{align*}
Thus, the constraint in~\eqref{eq:LP-level} is equivalent to $\bH \by \leq \ones$.

Similarly, to write the constraint in~\eqref{eq:LP-decode}, we define a matrix $\bG\in \mathbb{R}^{K\binom{K-1}{t}\times \ml}$. %
The columns of $\bG$ are labeled similar the entries of $\by$, and each row in $\bG$ is labeled by a pair $(k,\cT)$ where 
$\cT\subseteq [K]\setminus \{k\}$, and $|\cT|=t$. An entry at row $(k,\cT)$ and column $(\ell,\cS)$ is given by 
\begin{align*}
    \bG_{(k,\cT),(\ell,\cS)} &=  
    \begin{cases}
    -\overline{F}_{L_k}(\ell)  & \textrm{if $\cS = \cT \cup \{k\}$},\\
    0 & \textrm{otherwise},
    \end{cases}
\end{align*}
and the entries in the $m$th column are
\begin{align*}
    \bG_{(k,\cT),m} &= \frac{1}{\binom{K}{t}}, \qquad \forall (k,\cT). 
\end{align*}
With this, the constraint~\eqref{eq:LP-decode} is reduced to $\bG \by \leq \zeros$. 

We can conclude the following proposition by rephrasing the coding scheme devised above and its constraints in a linear form.  

\begin{proposition}\label{prop:LP}
For any  $K$-DTVBC with cache placement strategy $\cache$ and a (distinct) request profile $\bd$, the source rate $\fl$ given by 
\begin{align}\label{eq:LP_ach_schm}
    \fl = & \min \ [\zeros_{m-1}^T , -1] \by\\
    & \ \textrm{s.t.} \quad \bG \by \leq \zeros\nonumber \\
    &\phantom{\textrm{s.t.} \quad} \ \bH \by \leq \ones\nonumber \\ 
    &\phantom{\textrm{s.t.} \quad} -\by \leq \zeros\nonumber,
\end{align}
is achievable. 
\end{proposition}

In the following example, we evaluate $\fl$ by solving~\eqref{eq:LP_ach_schm} for a non-degraded broadcast channel with $K=3$ users.  We also solve the LP in~\eqref{eq:fwKUser} and show that the achievable rate and the upper bound do not match. This shows that our result does not provide an exact characterization for the maximum source rate, when the channel is not degraded, and the number of users is more than $2$. 
\begin{example}[\textit{Non-Degraded Channel Case}]\label{ex:LP-mu=1/3}
Consider  a network with $K=3$ users $N=3$ files, namely $\{W_1,W_2,W_3\}$, and $\nB=3$ transmit levels. The channel statistics of the three users are given by the CCDFs as
	\begin{align*}
	&\overline{F}_{L_1} =[0.9,0.3,0.3]^T, \cr
	&\overline{F}_{L_2} = [0.7,0.4,0.4]^T,\cr
	&\overline{F}_{L_3} = [0.5,0.5,0.5]^T.
	\end{align*} 
    Consider the caching strategy $\cent$ with $\mu = 1/3$, i.e.,
    \begin{align*}
        \bc_1 = (0,1/3],\quad \bc_2 = (1/3, 2/3], \quad \bc_3 = (2/3,1].   
    \end{align*}  
    Again,  we assume that user $k$ is interested in file $W^{(n)}_k$, for $k\in\{1,2,3\}$.  
    Here, the LP in~\eqref{eq:LP_ach_schm} is given by
    \begin{align}\label{eq:LP_ach_ex}
    \fl = & \min \ [\zeros_{1\times (m-1)} , -1] \by\\
    &\  \textrm{s.t.} \quad \bG \by \leq \zeros\nonumber \\
    &\ \phantom{\textrm{s.t.} \quad} \bH \by \leq \ones\nonumber \\ 
    &\phantom{\textrm{s.t.} \quad} -\by \leq \zeros\nonumber,
\end{align}
\addtocounter{equation}{3}
where the matrices $\bG$ and $\bH$ are given in~\eqref{eq:mat_G} and~\eqref{eq:mat_H}, at the top of this page. 
   
    The optimum solution of the LP in~\eqref{eq:LP_ach_ex} is also presented in~\eqref{eq:vect_y}
 where $\by_{(\ell,\cS)}$ indicates the fraction of time that the transmitter uses signal level $\ell$ to send a coded message $W^{(n)}_{\cS}$. The transmitter has to send coded messages $W^{(n)}_{\{1,2\}} = W^{(n)}_{1,\{2\}} \oplus W^{(n)}_{2,\{1\}}$, $W^{(n)}_{\{1,3\}} = W^{(n)}_{1,\{3\}} \oplus W^{(n)}_{3,\{1\}}$, and $W^{(n)}_{\{2,3\}} = W^{(n)}_{2,\{3\}} \oplus W^{(n)}_{3,\{2\}}$. Note that the source rate of ${\fl =\by_m= 3/2}$ is achievable. 
 Therefore, since the rate of each coded message is $1/3$ of the rate of the original files, the source rate for each coded message is $\fl/3 = 1/2$.
	
The transmission scheme devised by~\eqref{eq:vect_y} suggests that $W^{(n)}_{\{1,2\}}$ will be broadcast over the top level ($\ell=1$) for $2/3$ fraction of time, and the second level ($\ell=2$) for $1/12$ fraction. Thus, user~$1$ is able to decode the message $W^{(n)}_{\{1,2\}}$, as
	\begin{align*}
	    0.9 \times \frac{2}{3} + 0.3 \times \frac{1}{12} = \frac{5}{8} \geq \frac{1}{2} = \frac{\fl}{3}.
	\end{align*}
	Similarly, user $2$ decodes the message, since 
	\begin{align*}
	   0.7 \times \frac{2}{3} + 0.4 \times \frac{1}{12} = \frac{1}{2} \geq \frac{1}{2} = \frac{\fl}{3}.
	\end{align*} 
    A similar argument holds for decodability of $W^{(n)}_{\{1,3\}}$ at users $1$ and $3$, as well as decodability of $W^{(n)}_{\{2,3\}}$ at users $2$ and $3$. 
    
    Next, we evaluate the 
    upper-bound $\fs(\cent,\bd)$. We solve the LP problems in~\eqref{eq:fwKUser}  
     for all possible permutations. The bound corresponding to each permutation is given in Table~\ref{tab:f_perm}.
		\begin{table}[h]
	    \centering
	    \begin{tabular}{|c|c|}
	    \hline
	    $\pi$ & $\fs_{\pi}(\cent,\bd)$ \\
	    \hline\hline
	    $(3,2,1)$ & $1.66$ \\
	   \hline
	    $(3,1,2)$ & $1.76$ \\
	   \hline
	   $(2,3,1)$ & $\textbf{1.61}$ \\
	   \hline
	   $(2,1,3)$ & $1.62$ \\
	   \hline
	   $(1,3,2)$ & $1.73$ \\
	   \hline
	   $(1,2,3)$ & $1.64$ \\
	   \hline
	   \end{tabular}
    \vspace{5pt}
	   \caption{The upper bound on the source rate for each permutation with the caching strategy $\cent$ and the normalized cache size $\mu=1/3$.}
	    \label{tab:f_perm}
	\end{table}
 
    Therefore, we get
    \begin{align*}
        \htf & :=\fs(\cent,\bd)= \min_{\pi\in \Pi}\fs_{\pi}(\cent,\bd)\\
        & = \min \{1.66,    1.76, \textbf{1.61}, 1.62,    1.73,    1.64\} = 1.61,
	\end{align*}
    where the minimum value is attained for the permutation ${\pis = (2,3,1)}$ with ${\bfw^\s=(0,1.25,1)}$. Clearly, we have ${\fl = 1.5 < 1.61 = \htf} $, and there is a gap between the achievable rate and the upper bound. \hfill $\diamond$
\end{example}
\section{Proof of Theorem~\ref{thm:2-User}}\label{sec:proof_main_2}
In this section, we present the proof of Theorem~\ref{thm:2-User}. To do this, we first provide the converse proof and then discuss the achievability part. For the converse proof, we first enhance the channel to achieve a degraded broadcast channel. Next, we introduce a lemma that allows us to establish an upper bound on the maximum achievable source rate for a physically degraded BC. By applying this lemma to the degraded channel we have obtained, we are able to characterize the maximum achievable source rate.
\subsection{Converse}
The converse proof of Theorem~\ref{thm:2-User} is based on the result of~\cite{david2012fading}. We need to construct a degraded broadcast channel. In this regard, we replace $L_2$, the channel of User 2, with an enhanced channel~$\widetilde{L}_2$. For a given $\w\geq 1$, we define 
\begin{align}\label{eq:enh-2user}
    \ov{F}_{\wdt{L}_2}(\ell) \!:=\! \min \lb 1, \max \lb \ov{F}_{L_2}(\ell), \w \ov{F}_{L_1}(\ell) \rb\rb, \!\quad \ell\in [B],
\end{align}
which is the CCDF of random variable $\widetilde{L}_2$. Moreover, we define ${\wdt{Y}_2:=D^{B-\wdt{L}_2}X = X(1:\wdt{L}_2)}$. Hence, we have a degraded broadcast channel, i.e., $X\leftrightarrow \wdt{Y}_2\leftrightarrow Y_{1}$.

Next, we derive an upper bound on the maximum achievable source rate for a physically degraded BC. We refer to the appendix for the proof of Lemma~\ref{lm:2user-deg-BC}.

\begin{lm}\label{lm:2user-deg-BC}
    Consider a physically degraded memoryless BC described by $P_{Y_1,Y_2|X}$ for a given cache placement strategy $\cache$ and a distinct request profile $\bd$. Then, any achievable  source rate $f(\cache,\bd)$ satisfies 
	\begin{align}
	& f(\cache,\bd)\leq \frac{ I(U_1;Y_{1})}{1-\mu},\label{eq:f-I-2User}\\
	& f(\cache,\bd)\leq  \frac{I(X;Y_2|U_1)}{1-2\mu}.\label{eq:f-I-2User-2}
	\end{align}
for some $U_1$ satisfying $U_1\leftrightarrow X\leftrightarrow Y_2\leftrightarrow Y_{1}$.
\end{lm}
Now, we are ready for the proof of Theorem~\ref{thm:2-User}. We can use Lemma~\ref{lm:2user-deg-BC} for the degraded channel obtained by the enhancement procedure in~\eqref{eq:enh-2user}. Let $U_1$ be the random variable satisfying the claim of Lemma~\ref{lm:2user-deg-BC}, and form a Markov chain $U_1\leftrightarrow X\leftrightarrow \wdt{Y}_2\leftrightarrow Y_{1}$. Using the fact that CSI is available at the receivers, for the terms in~\eqref{eq:f-I-2User} and~\eqref{eq:f-I-2User-2}, we can write 
\begin{align}\label{eq:I-U_Y1}
     I(U_1;Y_1,L_1)&= I(U_1;X(1:L_1),L_1) \nonumber\\
    & = \sum_{j=1}^{B}\P_{L_1}(j)I(U;X(1:j),L_1=j)\nonumber\\
    & = \sum_{j=1}^{B}\lb \P_{L_1}(j)\sum_{\ell=1}^{j}I(U_1;X(\ell)|X(1:\ell-1)) \rb\nonumber\\
    & =\sum_{\ell=1}^{B}\lb I(U_1;X(\ell)|X(\ell-1))\sum_{j=1}^{\ell}\P_{L_1}(j)\rb\nonumber \\
    & =\sum_{\ell=1}^{B}\ov{F}_{L_1}(\ell)I(U_1;X(\ell)|X(\ell-1))\nonumber\\
    & = \sum_{\ell=1}^{B}\ov{F}_{L_1}(\ell)\bigg[ H(X(\ell)|X(1\!:\!\ell\!-\!1))- H(X(\ell)|X(1\!:\!\ell\!-\!1)),U_1)\bigg].
        \end{align}
Similarly, we have
\begin{align}\label{eq:I-U_Y2}
    I(X;\wdt{Y},\wdt{L}_2&|U_1)= I(X;X(1:\wdt{L}_2),\wdt{L}_2|U_1) \nonumber\\
    & = I(X;\wdt{L}_2|U_1) + I(X;X(1:\wdt{L}_2)|U_1,\wdt{L}_2)\nonumber \\
    & \stackrel{\rm (a)}{=} \sum_{j=1}^{B}\P_{\wdt{L}_2}(j)I(X;X(1:j)|U_1,\wdt{L}_2=j)\nonumber\\
    & = \sum_{j=1}^{B}\lb \P_{\wdt{L}_2}(j)\sum_{\ell=1}^{j}I(X;X(\ell)|X(1:\ell-1),U_1) \rb\nonumber\\
    & =\sum_{\ell=1}^{B}\lb I(X;X(\ell)|X(\ell-1),U_1)\sum_{j=1}^{\ell}\P_{\wdt{L}_2}(j)\rb\nonumber \\
    & =\sum_{\ell=1}^{B}\ov{F}_{\wdt{L}_2}(\ell)I(X;X(\ell)|X(\ell-1),U_1)\nonumber\\
    & = \sum_{\ell=1}^{B}\ov{F}_{\wdt{L}_2}(\ell)\big[H(X(\ell)|X(1:\ell-1),U_1) -H(X(\ell)|X,X(1:\ell-1)),U_1)\big]\nonumber\\
    & = \sum_{\ell=1}^{B}\ov{F}_{\wdt{L}_2}(\ell) H(X(\ell)|X(1:\ell-1),U_1),
\end{align}
where $\rm{(a)}$ follows from the fact that $\wdt{L}_2$ is independent from $U_1$ and $X$. Therefore, from Lemma~\ref{lm:2user-deg-BC}, we get
\begin{align}
    & \w (1-\mu) f(\cache,\bd) \!\leq\!\w \!\sum_{\ell=1}^{B}\ov{F}_{L_1}(\ell)\!\big[ H(X(\ell)|X(1\!:\!\ell\!-\!1)) -H(X(\ell)|X(1\!:\!\ell\!-\!1)),U_1)\big],\label{eq:f-w-2User}
 \end{align}
 and
 \begin{align}\label{eq:f-w-2User:0}
    & (1-2\mu) f(\cache,\bd) \leq \sum_{\ell=1}^{B}\ov{F}_{\wdt{L}_2}(\ell) H(X(\ell)|X(1:\ell-1),U_1).
\end{align}

Taking the sum of the two inequalities in~\eqref{eq:f-w-2User} and~\eqref{eq:f-w-2User:0}, we arrive at
\begin{align}\label{eq:fw-2User}
    &\!(\w(1\!-\!\mu) \!+\! (1\!-\!2\mu)) f(\cache,\bd) \nonumber\\
    & \!\leq\! \w \sum_{\ell=1}^B \ov{F}_{L_1}(\ell)H(X(\ell)|X(1\!:\!\ell\!-\!1)) \nonumber\\
    &\phantom{\leq} \!+\! \sum_{\ell=1}^{B}(\tilde{g}(\ell)\!-\!\w)\ov{F}_{L_1}(\ell)H(X(\ell)|X(1\!:\!\ell\!-\!1),U_1),
\end{align}
where $\tilde{g}(\ell):=\ov{F}_{\wdt{L}_2}(\ell)/\ov{F}_{L_1}(\ell)$. The summands of the first summation in~\eqref{eq:fw-2User} will be maximized by an i.i.d. Bernoulli random variable choice for $X_1,\ldots, X_{B}$. Moreover, the terms in the second summation can be maximized if
\begin{align*}
   H(X(\ell)|X(1\!:\!\ell\!-\!1),U_1) = 
    \begin{cases}
    1 & \tilde{g}(\ell)>\w,\\
    0 & \tilde{g}(\ell)\leq \w,
    \end{cases}
\end{align*}
which can be satisfied by an optimum choice for $U_1$, given by
\[U_1 = \{X(\ell)| \tilde{g}(\ell)\leq \w\}. \]
Hence, for~\eqref{eq:fw-2User} we can write
\begin{align}\label{eq:fw-2User-1}
    & \!(\w(1\!-\!\mu) \!+\! (1\!-\!2\mu)) f(\cache,\bd) \nonumber \\
    &\leq \w \sum_{\ell:\tilde{g}(\ell)\leq \w} \ov{F}_{L_1}(\ell) + \sum_{\ell:\tilde{g}(\ell)>\w}\tilde{g}(\ell)\ov{F}_{L_1}(\ell)\nonumber\\
    & \stackrel{\rm{(a)}}{=} \w \sum_{\ell:\tilde{g}(\ell)\leq \w} \ov{F}_{L_1}(\ell) + \sum_{\ell:\tilde{g}(\ell)>\w}\ov{F}_{\wdt{L}_2}(\ell)\nonumber\\
    & \stackrel{\rm{(b)}}{=}  \w \sum_{\ell:\tilde{g}(\ell)\leq \w} \ov{F}_{L_1}(\ell) + \sum_{\ell:\tilde{g}(\ell)>\w}\ov{F}_{L_2}(\ell)\nonumber\\
    & = \w R_1(\w)+R_2(\w),
\end{align}
where $R_1(\w)$ and $R_2(\w)$ are defined in~\eqref{eq:R1-R2}. We note that in the chain of inequalities in~\eqref{eq:fw-2User-1}, the step $\rm{(a)}$ follows from ${\tilde{g}(\ell)\ov{F}_{L_1}(\ell) = \ov{F}_{L_2}(\ell)}$ and ${\{\ell|\tilde{g}(\ell)>\w\} = \{\ell|g(\ell)>\w\}}$ and~$\rm{(b)}$ holds since $\ov{F}_{\wdt{L}_2}(\ell) = \ov{F}_{L_2}(\ell)$ whenever $g(\ell)>\w$.
Dividing both sides of~\eqref{eq:fw-2User-1} by $\w(1\!-\!\mu) + (1\!-\!2\mu)$ and minimizing over all $\w\geq 1$, we arrive at the the first minimization in~\eqref{eq:f-2User}.

For $0\leq \w \leq 1$, we can repeat the steps in~\eqref{eq:enh-2user} through~\eqref{eq:fw-2User-1}
by swapping the labels of the users and replacing $\w$ by $\frac{1}{\w}$. Under these reversed labels, we now enhance the channel of User 1 and get the second minimization in~\eqref{eq:f-2User}.

Finally, we characterize the maximum achievable source rate when $\mu>\frac{1}{2}$. Starting~\eqref{eq:f-I-2User} and using~\eqref{eq:I-U_Y1}, we can write
\begin{align}\label{eq:f-mu-2}
    (1-\mu) f(\cache,\bd)  & \leq I(U_1;Y_1,L_1)\nonumber\\
    & = \sum_{\ell=1}^{B}\ov{F}_{L_1}(\ell)\big[ H(X(\ell)|X(1:\ell-1)) -H(X(\ell)|X(1\!:\!\ell\!-\!1)),U_1)\big]\nonumber\\
    & \leq \sum_{\ell=1}^{B}\ov{F}_{L_1}(\ell).
\end{align}
Similarly, by swapping the labels of the users, we can repeat the steps in~\eqref{eq:f-mu-2} that leads us to
\begin{align*}
    (1-\mu)f(\cache,\bd) \leq \sum_{\ell=1}^{B}\ov{F}_{L_2}(\ell).
\end{align*}
This completes the converse proof.

\subsection{Achievability}
     The achievability proof of Theorem~\ref{thm:2-User} is based on a linear scheme in which the transmitter only broadcasts a raw and linear combination of the messages. To show that $\fs$ is achievable, we split the messages into three messages, including two private messages (one for each user) and a common message, which is intended for both users. Then, we allocate the (signal) levels in $[B]$ to each of these messages and prove that both users can decode their desired file. 
     
     Without loss of generality, we assume the source rate $\fs$ is obtained by the first minimization in~\eqref{eq:f-2User} and~\eqref{eq:f-2User-2} for $0\leq \mu \leq\frac{1}{2}$ and $\frac{1}{2}\leq\mu \leq 1$, respectively. Now, we present the achievability proof for each regime of $\mu$.
     
     \noindent\underline{$0\leq \mu \leq\frac{1}{2}:$}\\[2mm]
     Consider the central cache placement of Definition~\ref{def:central}. Let us denote the file requested by user $k$  by $W_{d_k}$ for $k\in[2]$. Recall that user $1$ needs $W_{d_1}^{(n)}$, which is partitioned into ${\left(W_{d_1,\emptyset}^{(n)}, W_{d_1,\{1\}}^{(n)}, W_{d_1,\{2\}}^{(n)}, W_{d_1,\{1,2\}}^{(n)}\right)}$. Note that since $\mu\leq \frac{1}{2}$, from ~\eqref{eq:cent-intersection} we can conclude that 
     $|\ccent{\{1,2\}}|=0$, and hence ${W_{d_1,\{1,2\}}^{(n)}=\emptyset}$.   While user $1$ has $W_{d_1,\{1\}}^{(n)}$ in its cache, the subfiles 
    $W_{d_1,\emptyset}^{(n)}$ and $W_{d_1,\{2\}}^{(n)}$, need to be delivered. Similarly, the  user $2$ will be served be by sending $W_{d_2,\emptyset}^{(n)}$, $W_{d_2,\{1\}}^{(n)}$. Instead of sending these message separately,  we send \emph{individual} messages  $W_{d_1,\emptyset}^{(n)}$ and $W_{d_2,\emptyset}^{(n)}$, as well as the \emph{common} message $W_{d_1,\{2\}}^{(n)} \oplus W_{d_2,\{1\}}^{(n)}$. 
    We aim to send each individual message to the intended user and the common message to both receivers. To this end, we need to allocate the \emph{levels} and \emph{time} among the messages. 
     We first define $g(\ell):= \overline{F}_{L_2}(\ell)/\overline{F}_{L_1}(\ell)$ for each level $\ell\in [B]$, and sort all the $B$ levels of the channel in an non-decreasing order according to $g(\cdot)$. This leads to a one-to-one mapping $\lambda :[B]\rightarrow [B]$ that sorts the level, and thus, ${g(\lambda(1)) \leq g(\lambda(2)) \leq \cdots \leq g(\lambda(B-1)) \leq g(\lambda(B))}$. For notational simplicity, we rename the levels and define $\el_i:=\lambda(i)$ and $\gamma_{i}: = g(\el_i)$,  for every $i \in [B]$. We clearly have ${\gamma_1 \leq \gamma_2\leq \cdots \leq \gamma_B}$. We refer to Figure~\ref{fig:u_v_star} for clarification. Our proposed level (and time) allocation scheme is parameterized by $(u,v; \alpha, \beta)$, where $u,v\in [B]$ with $u\leq v$, and $0<\alpha, \beta\leq 1$: We use levels $\{\ell_1,\ell_2\dots, \ell_{u-1}\}$ for the entire communication block and level $\ell_u$ for an $\alpha$ fraction of time to send the individual message $W_{d_1,\emptyset}^{(n)}$. Similarly, the individual message $W_{d_2,\emptyset}^{(n)}$ will be sent on levels $\{\ell_{v+1}, \dots, \ell_B\}$ for the entire communication block and on level $\ell_v$ for $\beta$ fraction of time. The remaining levels (including the remaining $(1-\alpha)$ fraction of $\ell_u$ and $(1-\beta)$ fraction of $\ell_v$) will be used to send the common message. Such a delivery strategy can support any source rate $f$ that satisfies 
\begin{align}\label{eq:cent-2:ind-com}
    &\sum_{i<u} \ov{F}_{L_1}(\el_i) +\alpha \ov{F}_{L_1}(\el_u) \geq \frac{1}{n}\left|W_{d_1,\emptyset}^{(n)}\right| = (1-2\mu) f,\nonumber\\
    &\sum_{i>v} \ov{F}_{L_2}(\el_i) +\beta \ov{F}_{L_2}(\el_v) \geq \frac{1}{n}\left|W_{d_2,\emptyset}^{(n)}\right| = (1-2\mu) f,\nonumber\\
    &
    (1-\alpha) \ov{F}_{L_1}(\el_u) + \sum_{u<i<v}
    \ov{F}_{L_1}(\el_i) +(1-\beta) \ov{F}_{L_1}(\el_v)  \geq \frac{1}{n}\left|W_{d_1,\{2\}}^{(n)} \oplus W_{d_2,\{1\}}^{(n)}\right| = \mu f,\\
    &
    (1-\alpha) \ov{F}_{L_2}(\el_u) + \sum_{u<i<v}
    \ov{F}_{L_2}(\el_i) +(1-\beta) \ov{F}_{L_2}(\el_v)  \geq \frac{1}{n}\left|W_{d_1,\{2\}}^{(n)} \oplus W_{d_2,\{1\}}^{(n)}\right| = \mu f\nonumber.
\end{align}
It is easy to verify that constraints in~\eqref{eq:cent-2:ind-com} are feasible if and only if the constraints 
\begin{align}\label{eq:cent-2:ind-sum}
\begin{split}
    &\sum_{i<u} \ov{F}_{L_1}(\el_i) +\alpha \ov{F}_{L_1}(\el_u) \geq  (1-2\mu) f,\\
    &
    (1-\alpha) \ov{F}_{L_2}(\el_u) + \sum_{i>u }
    \ov{F}_{L_2}(\el_i) \geq (1-\mu) f\\
    &\sum_{i>v} \ov{F}_{L_2}(\el_i) +\beta \ov{F}_{L_2}(\el_v) \geq (1-2\mu) f,\\
    &
    \sum_{i<v}
    \ov{F}_{L_1}(\el_i) +(1-\beta) \ov{F}_{L_1}(\el_v)  \geq (1-\mu) f,
    \end{split}
\end{align}
are satisfied. 
Note that we can optimize the allocation parameters $(u,v; \alpha, \beta)$. Moreover, the first and the second constraints in~\eqref{eq:cent-2:ind-sum} only depend on $(u,\alpha)$, and the third and the fourth constraints only depend on $(v,\beta)$. These motivate defining 
\begin{align}
    & \! f_1(u,\alpha) \!: =\!  \min \left(\!\frac{1}{1\!-\!2\mu}\left[\sum_{i<u} \ov{F}_{L_1}(\el_i) \!+\!\alpha \ov{F}_{L_1}(\el_u) \right],\frac{1}{1\!-\!\mu}\!\left[\sum_{i>u} \ov{F}_{L_2}(\el_i) \!+\!(1\!-\!\alpha) \ov{F}_{L_2}(\el_u) \right]\right)\!, \label{eq:f_1_u_a}\\
    & \! f_2(v,\beta) \!: =\! \min \left(\!\frac{1}{1-\mu}\left[\sum_{i<v} \ov{F}_{L_1}(\el_i) + (1-\beta) \ov{F}_{L_1}(\el_v) \right],\frac{1}{1-2\mu}\left[\sum_{i>v} \ov{F}_{L_2}(\el_i) \!+\! \beta \ov{F}_{L_2}(\el_v) \right] \right)\!.\label{eq:f_2_v_b}
\end{align}
Our goal would be to maximize $\min(f_1(u,\alpha),f_2(v,\beta))$. The following lemma formally presents the properties of the optimum solution of $f_1(u,\alpha)$ and $f_2(v,\beta)$. We show that the maximum of $\min(f_1(u,\alpha),f_2(v,\beta))$ over the choice of $(u,v;\alpha,\beta)$ meets the upper bound of source rate in~\eqref{eq:f-2User}. This completes the achievability proof for the regime $0\leq\mu\leq \frac{1}{2}$

\begin{lm}\label{lm:2user_f_ach}
Consider a $2$-DTVBC with a distinct request profile $\bd$, a normalized cache size $\mu\leq \frac{1}{2}$, and the central caching strategy $\cent$. Let $\fs_1 := \max_{u,\alpha} f_1(u,\alpha)$ and $\fs_2 := \max_{v,\beta}f_2(v,\beta)$, where $f_1(u,\alpha)$ and $f_2(v,\beta)$ are defined in~\eqref{eq:f_1_u_a} and~\eqref{eq:f_2_v_b}, respectively. Then, the following properties hold:
\begin{enumerate}
    \item[(i)] The source rate $ \min(f^\s_1,f^\s_2)$ is achievable;
    \item[(ii)] If ${(u^\s,\alpha^\s):=\arg\max f_1(u,\alpha)}$ be the maximizer of $f_1$ and   ${(v^\s,\beta^\s):=\arg\max f_2(v,\beta)}$ be the maximizer of $f_2$,  then we have $u^\star \leq v^\star$;
    \item[(iii)] If $f^\s_1\leq f^\s_2$  then $g(\ell_{u^\s}) \leq 1$. Alternatively, if $f^\s_1\geq f^\s_2$, then we have $g(\ell_{v^\s}) \geq 1$; 
    \item[(iv)] For $\fs$ defined in~\eqref{eq:f-2User}, we have  $\fs \leq \min(f^\s_1,f^\s_2)$. 
\end{enumerate}
\end{lm}
The proof of Lemma~\ref{lm:2user_f_ach} is presented in the appendix.
It is worth noting that parts (i) and (iv) of the lemma above immediately yield the achievability proof of Theorem~\ref{thm:2-User} for $0\leq \mu \leq \frac{1}{2}$.

\begin{figure}
    \centering
    \includegraphics[width=0.5\linewidth]{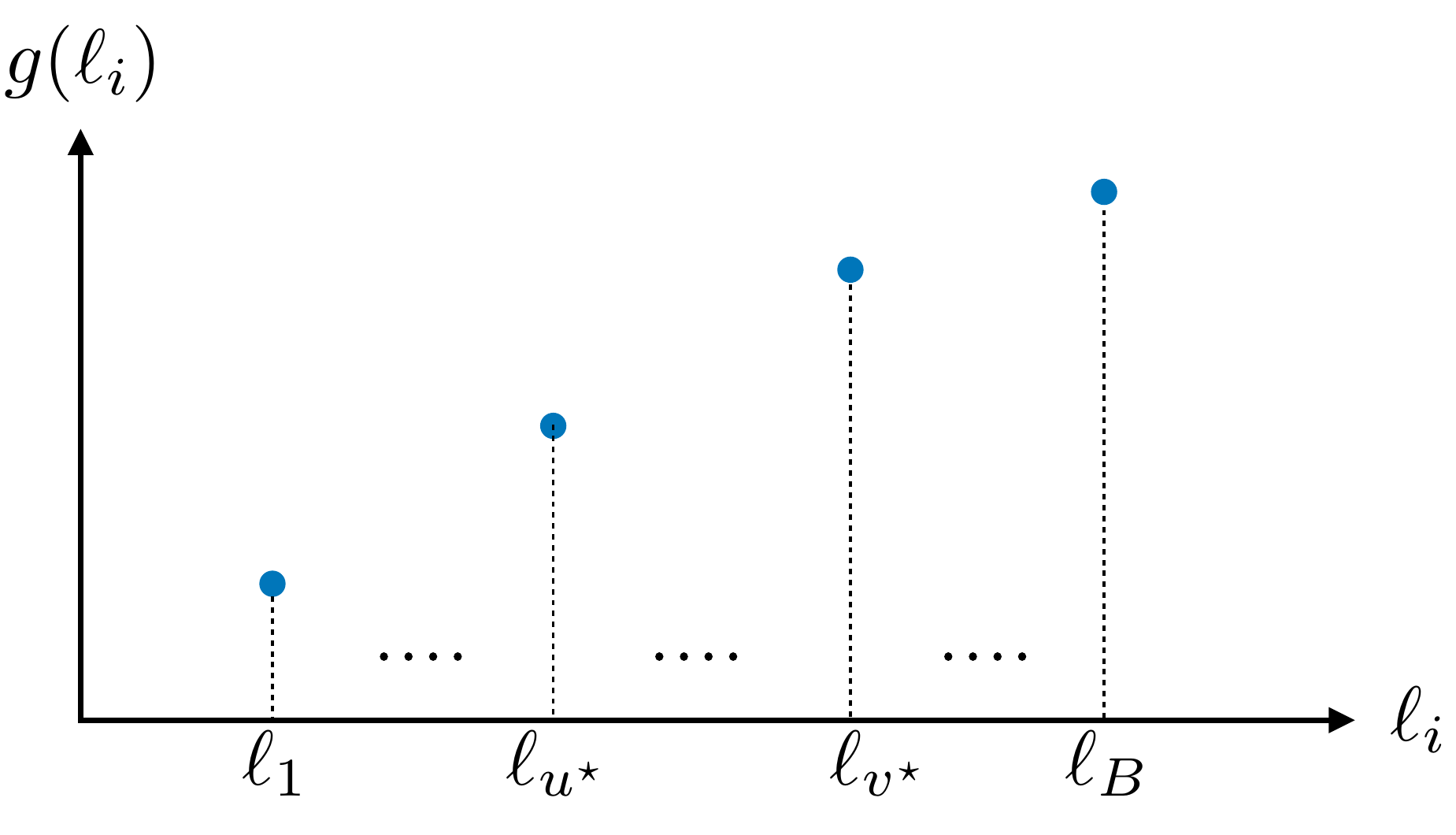}
    \caption{Sorting the signal levels of a 2-user system according to their ratio. }
    \label{fig:u_v_star}
\end{figure} 

  \vspace{5mm}
    \noindent\underline{$\frac{1}{2}\leq \mu \leq 1:$}\\[2mm]
    We need to show that the minimum attained in~\eqref{eq:f-2User-2} is achievable. Let  $\bd =(d_1,d_2)$ be the demand profile. The file $W_{d_1}^{(n)}$ is partitioned into ${\left(W_{d_1,\emptyset}^{(n)}, W_{d_1,\{1\}}^{(n)}, W_{d_1,\{2\}}^{(n)}, W_{d_1,\{1,2\}}^{(n)}\right)}$. Note that since $\mu\geq \frac{1}{2}$, from ~\eqref{eq:cent-intersection} we can conclude that 
     $|\ccent{\emptyset}|=0$, and hence ${W_{d_1,\emptyset}^{(n)}=\emptyset}$. This means user $1$ has $W_{d_1,\{1\}}^{(n)}$ and  $W_{d_1,\{1,2\}}^{(n)}$ in its cache, and only $W_{d_1,\{2\}}^{(n)}$ need to be delivered over the channel. Similarly, user $2$ will be served by $W_{d_2,\{1\}}^{(n)}$. The transmitter only needs to multicast a common message $W_{d_1,\{2\}} \oplus W_{d_2,\{1\}}$ to both users over the signal levels in $[B]$. The size of this common message is $\left(1-\mu\right)\fs$. Hence, the maximum achievable source rate is given the minimum of the capacities of the channels to two users. For user~$i$, the maximum rate is 
        \begin{align*}
            \fs = \frac{\sum_{\ell=1}^{B}\ov{F}_{L_i}(\ell)}{1-\mu}.
        \end{align*}
        Taking the minimum over $i\in\{1,2\}$, we get the rate in~\eqref{eq:f-2User-2}.   
	This completes the proof of the theorem.
	\hfill $\square$ 
	
\section{Achievability Proof of Theorem~\ref{thm:degraded}}\label{sec:proof-degraded-ach}

The achievability proof of Theorem~\ref{thm:degraded} is based on the achievable scheme in Section~\ref{sec:LP} and Proposition~\ref{prop:LP}.  
We consider a specific normalized cache sizes $\mu$ for which  ${t= K\mu =KM/N \in \mathbb{N}}$. Therefore, from~\eqref{eq:cent-union}, we have
\begin{align}\label{eq:m_k_t}
    \tcents_{[k]} = \frac{\binom{K-k}{t}}{\binom{K}{t}},
\end{align}
for every $k\in [K]$.
Now, we show that every set $\{z_{\ell,k}\}$ that satisfies~\eqref{eq:const_thm2} provides a feasible set of $\{y_{\ell,\cS}\}$ for the achievable scheme of Proposition~\ref{prop:LP}. 
Recall that we assume the channel is degraded (i.e., $L_K\!\geq_{\mathsf{st}}\!\cdots \geq_{\mathsf{st}} L_1$). Let $\kw := \min_{k\in \cS} k$ be the index of the \emph{weakest} user in the set $\cS$. For every $\cS\subseteq [K]$ with $|\cS|=t+1$, we set
\begin{align}\label{eq:y-in-z}
     y_{\ell,\cS} = \frac{1}{\binom{K-\kw}{t}}z_{\ell,\kw}.
\end{align}
First, note that $y_{\ell,\cS}\geq 0$ for every $\ell \!\in\! [B]$.

Next, we have
\begin{align}\label{eq:const_2}
  \sum_{\substack{\cS\subseteq [K]\\|\cS|= t+1 }}  y_{\ell,\cS}  & \stackrel{\rm{(a)}}{=} \sum_{\substack{\cS\subseteq [K]\\|\cS|= t+1 }} \frac{z_{\ell,\kw}}{\binom{K-\kw}{t}}\nonumber\\
  & = \sum_{k=1}^{K} \sum_{\substack{\cS\subseteq [K]\\|\cS|= t+1 \\ \kw=k}} \frac{z_{\ell,k}}{\binom{K-k}{t}}\nonumber\\
    & = \sum_{k=1}^{K} \frac{z_{\ell,k}}{\binom{K-k}{t}}\sum_{\substack{\cS\subseteq [K]\\|\cS|= t+1 \\ \kw=k}} 1\nonumber\\
    & \stackrel{\rm{(b)}}{=} \sum_{k=1}^{K} \frac{z_{\ell,k}}{\binom{K-k}{t}} \binom{K-k}{t} \nonumber\\
    & = \sum_{k=1}^{K} z_{\ell,k} 
     \stackrel{\rm{(c)}}{\leq} 1,
\end{align}
where $\rm{(a)}$ follows from~\eqref{eq:y-in-z}, in $\rm{(b)}$ we used the fact that any 
$\cS$ satisfying $\cS\subseteq [K]$, $|\cS|= t+1$ and $\kw=k$  should be of the form of $\cS = \{k\} \cup \cT$, where $\cT\subseteq \{k+1, k+2, \dots, K\}$, with $|\cT|=t$. Hence, the number of such $\cS$'s is $\binom{K-k}{t}$. Lastly, $\rm{(c)}$ follows from the last constraint in~\eqref{eq:const_thm2}. Thus, ~\eqref{eq:const_2} shows that the $\{y_{\ell, \cS}\}$ introduced above satisfies~\eqref{eq:LP-level}.

Furthermore, the degradedness of the channel implies that 
\begin{align}\label{eq:const_1}
    \sum_{\ell=1}^\nB \overline{F}_{L_k}(\ell) y_{\ell,\cS} & \stackrel{\rm{(a)}}{\geq} \sum_{\ell=1}^\nB \overline{F}_{L_{\kw}}(\ell) y_{\ell,\cS} \nonumber\\
    & = \frac{1}{\binom{K-\kw}{t}}\sum_{\ell=1}^\nB \overline{F}_{L_{\kw}}(\ell) z_{\ell,\kw}\nonumber\\
    &\stackrel{\rm{(b)}}{\geq} \frac{1}{\binom{K-\kw}{t}}  \lp 1-\tcents_{[\kw]}\rp \bar{f} \nonumber\\
    & \stackrel{\rm{(c)}}{\geq}  \frac{1}{\binom{K-\kw}{t}} \frac{\binom{K-\kw}{t}}{\binom{K}{t}} \bar{f}\nonumber\\
    & = \frac{\bar{f}}{\binom{K}{t}},
\end{align}
for every $k\in \cS$. 
Here, $\rm{(a)}$ holds since for every user ${k\in \cS}$ and every $\ell \in [B]$, we have $\overline{F}_{L_k}(\ell) \geq \overline{F}_{L_{\kw}}(\ell)$,
$\rm{(b)}$ follows from the first constraint in~\eqref{eq:const_thm2},  and we used~\eqref{eq:m_k_t} with ${k=\kw}$ in $\rm{(c)}$. This shows that the $\{y_{\ell,\cS}\}$ sequence introduced in~\eqref{eq:y-in-z} satisfies~\eqref{eq:LP-decode}. %

We proved that the $\{y_{\ell,\cS}\}$ introduced in~\eqref{eq:y-in-z} satisfies all constraints of the LP problem in Proposition~\ref{prop:LP}. In other words, every $\{z_{\ell,k}\}$ satisfying~\eqref{eq:const_thm2} provides a feasible solution $\{y_{\ell,\cS}\}$ for the (achievable) optimization method in~\eqref{eq:LP}. This, together with Proposition~\ref{prop:LP} (that any feasible solution of $\{y_{\ell,\cS}\}$ leads to an achievable source rate) completes the achievability proof of Theorem~\ref{thm:degraded}. 
\hfill $\square$ 

\section{Proof of Theorem~\ref{thm:Kuser-deg-BC} and Proposition~\ref{pro:LP-fbnd}}\label{sec:proof_main_1}
In this section, we provide the proof of Theorem~\ref{thm:Kuser-deg-BC} and Proposition~\ref{pro:LP-fbnd}. First, we present some auxiliary lemmas whose proofs are provided in the appendix.
\subsection{Preliminary results}
The proof of Theorem~\ref{thm:Kuser-deg-BC} is built based on the result of~\cite{yates2011k}, in which the rate region of an erasure $K$-user broadcast channel is characterized. We first need to enhance the channels to convert the network to a degraded broadcast channel. To this end, we will replace $L_k$, the channel of User~$k$, by a \emph{stronger} channel $\widetilde{L}_k$, so that the channel of user $k$ statistically degrades that of user $k-1$, for $k=2,3,\dots, K$. More precisely, for a given weight vector ${\bfw = (\omega_1,\ldots,\omega_K) \in [0,\infty)^K}$ with \textit{sorted} entries ${\w_1\geq \w_2\geq \cdots \geq \w_K}$, we define 
\begin{equation}\label{eq:chnl_ench}
	  \overline{F}_{\widetilde{L}_{k}}(\ell):=\min\left[1,\max\left(\overline{F}_{L_{k}}(\ell),\frac{\omega_{{k-1}}}{\omega_{k}}\overline{F}_{\widetilde{L}_{{k-1}}}(\ell)\right)\right],
\end{equation}
for every $\ell\in [B]$ and $k\in \{2,3,\dots, K\}$.
with an initialization given by $\overline{F}_{\widetilde{L}_{1}}(\ell) = \overline{F}_{{L}_{1}}(\ell)$,
for every $\ell \in [B]$.

The following lemma demonstrates some of the properties of the enhanced channel, which will be useful in the proof of Theorem~\ref{thm:Kuser-deg-BC}. 
\begin{lm}\label{lm:chnl_ench} 
The CCDF of $\wdt{L}_{k}$ providing in~\eqref{eq:chnl_ench} has the following properties
\begin{enumerate}
    \item[(i)] If $\overline{F}_{\widetilde{L}_{k}}(\ell) =1$, then
      \begin{align*}
          \overline{F}_{\widetilde{L}_{u}}(\ell) =1,
      \end{align*}
      for every $ u\geq k$. 
      \item[(ii)] If $\omega_{k} \overline{F}_{\widetilde{L}_{k}}(\ell)> \omega_{k-1}\overline{F}_{\widetilde{L}_{k-1}}(\ell)$, then
      \begin{align*}
         \w_k \overline{F}_{L_{k}}(\ell) \!=\! \w_k \overline{F}_{\wdt{L}_{k}}\!(\ell)  \!>\! \w_{k-1}\ov{F}_{\wdt{L}_{k-1}}\!(\ell) \!\geq\! \cdots \!\geq\! \w_{1}\ov{F}_{\wdt{L}_{1}}\!(\ell).
    \end{align*}
	   \item[(iii)] If $\omega_{k} \overline{F}_{\widetilde{L}_{k}}(\ell) < \omega_{k-1} \overline{F}_{\widetilde{L}_{k-1}}(\ell)$, then
	   \begin{align*}
	       \w_K \ov{F}_{\wdt{L}_K}(\ell) \leq \cdots \leq \w_k \ov{F}_{\wdt{L}_k}(\ell) < \w_{k-1} \ov{F}_{\wdt{L}_{k-1}}(\ell).
	   \end{align*}
	   \item[(iv)] The maximum of the weighted channel parameters satisfy
	     \begin{align*}
	      \max_k \w_k \ov{F}_{\wdt{L}_k}(\ell) = \max_k \omega_k \overline{F}_{{L}_{k}}(\ell),
	     \end{align*}
	     for every $\ell \in [B]$.
	     \end{enumerate}
\end{lm}  
The following corollary is based on the properties presented in Lemma~\ref{lm:chnl_ench} and provides a better understanding of the enhancement procedure. Note that for a given $\el \in [B]$, the quantity $\w_k \ov{F}_{\wdt{L}_k}(\ell)$ may be equal for $k\neq k'$. Hence, $\arg \max_k \w_k \ov{F}_{\wdt{L}_k}(\ell)$ is a set, with possibly many elements. 
\begin{corollary}\label{cor:enh_chnl}
  Let  ${k^\s := \min \arg \max_k \w_k \ov{F}_{\wdt{L}_k}(\ell)}$ and ${u^\s := \max \arg \max_k \w_k \ov{F}_{\wdt{L}_k}(\ell)}$ for any fixed level $\ell \in[B]$. Then, we can decompose the set of users $[K]$ into the following non-overlapping subsets
  \begin{align*}
      [K] = \{1,\ldots,k^\s\!-\!1\} \cup \{k^\s,\ldots,u^\s\} \cup \{u^\s\!+\!1,\ldots, K\}.
  \end{align*}
  Then, we arrive at the below properties 
  \begin{enumerate}
      \item[(i)] The sequence $\{\w_{k} \overline{F}_{L_{k}}(\ell)\}_{k=1}^{k^\s\!-\!1}$ is non-decreasing, i.e.,
      \begin{align*}
      \w_{1}\ov{F}_{\wdt{L}_{1}}(\ell) \leq \cdots \leq 
           \w_{k^\s-1}\ov{F}_{\wdt{L}_{k^\s-1}}(\ell)  <\w_{k^\s} \overline{F}_{\wdt{L}_{k^\s}}(\ell).
    \end{align*}
    \item[(ii)] The sequence $\{\w_{k} \overline{F}_{L_{k}}(\ell)\}_{k=k^\s}^{u^\s}$ satisfies 
    \begin{align*}
        \w_{k^\s} \overline{F}_{L_{k^\s}}(\ell) = \w_{k^\s} \overline{F}_{\wdt{L}_{k^\s}}(\ell) = \cdots = \w_{u^\s} \overline{F}_{\wdt{L}_{u^\s}}(\ell).
    \end{align*}
    \item[(iii)] The sequence $\{\w_{k} \overline{F}_{L_{k}}(\ell)\}_{k=u^\s\!+1}^{K}$ is non-increasing, i.e.,
    \begin{align*}
    \w_{u^\s} \ov{F}_{\wdt{L}_{u^\s}}(\ell) > \w_{u^\s\!+1} \ov{F}_{\wdt{L}_{u^\s\!+1}}(\ell) \geq \cdots \geq 
        \w_K \ov{F}_{\wdt{L}_K}(\ell).
    \end{align*}
  \end{enumerate}
\end{corollary}
The proof of Corollary~\ref{cor:enh_chnl} is presented in the appendix.
Using the properties discussed  in Corollary~\ref{cor:enh_chnl}, we can visualize the behavior of the enhanced channels $\ov{F}_{\wdt{L}_k}(\ell)$ and their weighted versions  $\w_k \ov{F}_{\wdt{L}_k}(\ell)$ as in Figure~\ref{fig:enh_ch}.
 \begin{figure}[t]
    \centering
    \includegraphics[width=0.5 \linewidth]{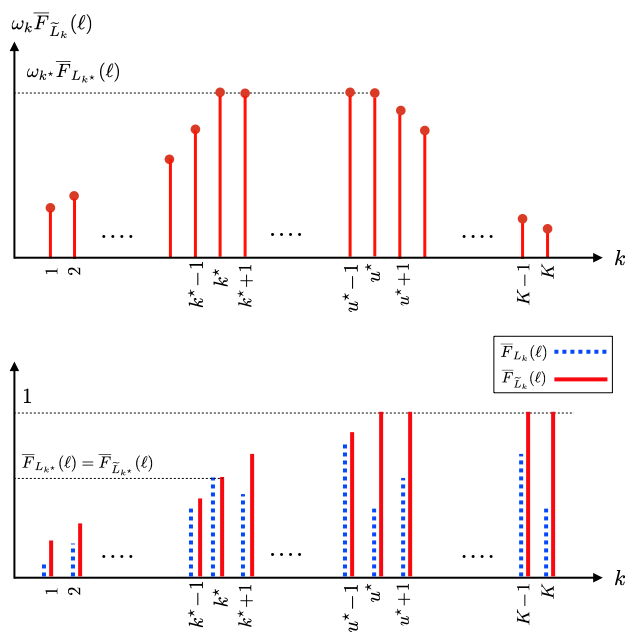}
    \caption{Channel enhancement: The behavior of the  $\w_k \ov{F}_{\wdt{L}_k}(\ell)$ for the enhanced deterministic broadcast channel (top), and comparison of the channel parameters before and after enhancement (bottom).}\label{fig:enh_ch}
\end{figure}

In the next lemma, we provide an important property of the cache placement strategy.
\begin{lm}\label{lm:I-mu}
For a given caching strategy and mutually independent files $W^{(n)}_1,\ldots,W^{(n)}_N$, we get
  \begin{align}\label{eq:I-f}
      I\lp W_i;C^{(n)}_\cS\rp \leq n\mu_{\cS}f,\quad \cS\subseteq[K].
  \end{align}
\end{lm}
Finally, we provide the extension of Lemma~\ref{lm:2user-deg-BC} to the $K$-users system in the following result.
\begin{lm}\label{lm:Kuser-deg-BC}
    If a source rate $f(\cache,\bd)$ for a given caching strategy $\cache$ and a distinct request profile $\bd$ is achievable on a physically degraded broadcast channel, i.e.,
	\[
	X\leftrightarrow Y_{K}\leftrightarrow\cdots\leftrightarrow Y_{1},
	\]
	then $f(\cache,\bd)$ satisfies 
	\begin{align}
	& f(\cache,\bd)\leq \frac{ I(U_{1};Y_{1})}{1-\mu_{[1]}}  ,\label{eq:lm:f_upper_1}\\
	& f(\cache,\bd) \leq  \frac{ I(U_{k};Y_{k}|U_{k-1})}{1-\mu_{[k]}},\quad k\in[2:K-1]\label{eq:lm:f_upper_k}\\
	& f(\cache,\bd)\leq  \frac{I(X;Y_{K}|U_{K-1})}{1-\mu_{[K]}},\label{eq:lm:f_upper_K}
	\end{align}
	for some random variables $(U_{1},U_{2},\ldots,U_{K-1})$ that form a Markov chain
	\[
	X\leftrightarrow U_{K-1}\leftrightarrow\cdots\leftrightarrow U_{1}.
	\]
\end{lm}
The proof of Lemma~\ref{lm:Kuser-deg-BC} is presented in Appendix.
\begin{remark}
The global capacity of a broadcast channel depends on the underlying transition probability  $\P(Y_1,\ldots,Y_K|X)$ only through its marginal conditional probabilities $\P(Y_1|X),\ldots,\P(Y_K|X)$. Therefore, the claim of Lemma~\ref{lm:Kuser-deg-BC} also applies to stochastically degraded BCs.
\end{remark}
Now, we are ready to present the proof of Theorem~\ref{thm:Kuser-deg-BC}.
\subsection{An upper-bound on the achievable source rate}
The main steps of the proof of Theorem~\ref{thm:Kuser-deg-BC} are twofold: 
We first enhance and replace the arbitrary $L_1,\ldots, L_K$ by the \textit{degraded} $\wdt{L}_1,\ldots,\wdt{L}_K$ and then by exploiting Lemma~\ref{lm:Kuser-deg-BC}, we derive an upper-bound on the achievable source rate.
\begin{proof}[Proof of Theorem~\ref{thm:Kuser-deg-BC}]
We first note that for arbitrary channels $L_1,\ldots, L_K$, the $K$-DTVBC is not degraded. Hence, we recursively enhance the channel of the users to obtain a set of degraded channels. In this regard, we consider a weight vector ${\bfw = (\omega_1,\ldots,\omega_K) \in [0,\infty)^K}$ with sorted entries ${\w_1\geq \w_2\geq \cdots \geq \w_K}$. We define the enhanced channel output of user $k$ as $\widetilde{Y}_{k}=D^{\nB-\widetilde{L}_{k}}X=X(1:\widetilde{L}_{k})$ where $\widetilde{L}_{k}$ is a random variable drawn according to $\overline{F}_{\widetilde{L}_{k}}$, given by~\eqref{eq:chnl_ench}, independent of all other users. Since $\w_{k-1} \geq \w_{k}$, from~\eqref{eq:chnl_ench} we have $\overline{F}_{\widetilde{L}_{k}}(\ell) \geq \overline{F}_{\widetilde{L}_{k-1}}(\ell)$. Thus, from~\cite[Lemma~1]{yates2011k}, we can conclude that the enhanced broadcast channel is (stochastically) degraded. 
    Now, we can use  Lemma~\ref{lm:Kuser-deg-BC} for the degraded channel obtained by the enhancement procedure. Let $U_1,\ldots, U_{K-1}$ be the random variables satisfying the claim of the lemma, and form a Markov chain ${X\leftrightarrow U_{K-1}\leftrightarrow\cdots\leftrightarrow U_{1}}$.  
    Then, given the fact that CSI is available at the receivers, the terms in Lemma~\ref{lm:Kuser-deg-BC} for the deterministic channel of interest will be simplified to 
\begin{align}\label{eq:I_bnd}
    & I(U_{k};\widetilde{Y}_{k},\widetilde{L}_{k}| U_{{k-1}}) \nonumber\\
    & = I(U_{k};X(1 : \widetilde{L}_{k}),\widetilde{L}_{k}|U_{{k-1}})\nonumber\\
	&= I(U_{k}; \widetilde{L}_{k}|U_{{k-1}})+ I(U_{k};X(1 : \widetilde{L}_{k}),|U_{{k-1}}, \widetilde{L}_{k})\nonumber\\
	&\stackrel{\rm (a)}{=}  \sum_{j=1}^\nB  \P_{\widetilde{L}_{k}}\!(j)  I(U_{k};X(1:j)|U_{{k-1}}, \widetilde{L}_{k} = j)\nonumber\\
	&= \sum_{j=1}^\nB \left[ \P_{\widetilde{L}_{k}}\!(j) \sum_{\ell=1}^j I(U_{k};X(\ell) | X(1:\ell-1), U_{{k-1}})   \right]\nonumber\\
	&= \sum_{\ell=1}^\nB \left[ I\left(U_{k};X(\ell) | X(1 :\ell-1), U_{{k-1}}\right)  \cdot \sum_{j=\ell}^\nB  \P_{\widetilde{L}_{k}}\!(j)\right]\nonumber\\
	&= \sum\nolimits_{\ell=1}^\nB \overline{F}_{\widetilde{L}_{k}}(\ell) I(U_{k};X(\ell) | X(1\!:\!\ell\!-\!1)|U_{{k-1}})\nonumber\\
	&=  \sum_{\ell=1}^\nB \overline{F}_{\widetilde{L}_{k}}\!(\ell) \left[H\left(X(\ell) | X(1\!:\! \ell\!-\!1),U_{{k-1}}\hspace{-1pt}\right)- H\left(X(\ell)|X(1\!:\! \ell\!-\!1) ,U_{{k-1}},U_{k}\right) \right]\nonumber\\
    &\stackrel{\rm(b)}{=}  \sum_{\ell=1}^\nB \overline{F}_{\widetilde{L}_{k}}\!(\ell) \left[H\hspace{-1pt}(\hspace{-1pt}X(\ell) | X(1\!:\! \ell\!-\!1\hspace{-1pt}),U_{{k-1}}\hspace{-1pt})- H\hspace{-1pt}(\hspace{-1pt}X(\ell)|X(1\!:\! \ell\!-\!1)\hspace{-1pt},U_{k}) \right]\nonumber\\
	& =  \sum_{\ell=1}^\nB \overline{F}_{\widetilde{L}_{k}}(\ell) Q_{\ell,{k}},
	\end{align}
    where $\rm{(a)}$ follows from the fact that the random variable $\widetilde{L}_{k}$ is independent from $U_{k-1}$ and $U_k$, $\rm{(b)}$ holds due to the Markov chain $U_{{k-1}}\leftrightarrow U_{k}\leftrightarrow X(\ell)$ and ${Q_{\ell,{k}} \!:=\! H\!\left(\hspace{-1pt}X(\ell) | X(1\!:\! \ell\!-\!1\hspace{-1pt}),U_{{k-1}}\right)\!-\! H\!\left(\hspace{-1pt}X(\ell)|X(1\!\!:\! \ell\!-\!1)\hspace{-1pt},U_{k}\right)}$. Therefore, from Lemma~\ref{lm:Kuser-deg-BC} we have
	\begin{align}\label{eq:lm:Kuser-deg-BC:simplified}
	\!f(\cache,\bd) \cdot \lp 1\!-\!\mu_{[k]} \rp & \leq\! I(U_{k}; \widetilde{Y}_{k}, \widetilde{L}_{k} | U_{k-1})\nonumber\\
    &=\!\sum_{\ell=1}^\nB \overline{F}_{\widetilde{L}_{k}}(\ell) Q_{\ell,{k}},
	\end{align}
	for $k\in [K]$. 
	Taking a weighted sum of~\eqref{eq:lm:Kuser-deg-BC:simplified} with coefficients $\{\w_k\}$, we arrive at
	\begin{align}\label{eq:thm2_0}
        f(\cache,\bd)\sum_{k=1}^{K}
	\omega_{k}\lp 1-\mu_{[k]} \rp
	& \leq \sum_{k=1}^{K}\sum_{\ell=1}^{\nB}\w_{k}\overline{F}_{\widetilde{L}_{k}}(\ell)\cdot Q_{\ell,{k}} \nonumber\\
        & = \sum_{\ell=1}^{\nB} \sum_{k=1}^{K}\w_{k}\overline{F}_{\widetilde{L}_{k}}(\ell)\cdot Q_{\ell,{k}}.
	\end{align}
	Note that for each $\ell\in[B]$ we have
	\begin{align}\label{eq:sum_Q}
	    & \sum_{k=1}^K Q_{\ell,{k}}\nonumber\\
            &= \!\!\sum_{k=1}^K \!\left[H\!\left(\hspace{-1pt}X(\ell) | X(1\!:\! \ell\!-\!1\hspace{-1pt}),U_{{k-1}}\right) \!-\! H\!\left(\hspace{-1pt}X(\ell)|X(1\!\!:\! \ell\!-\!1)\hspace{-1pt},U_{k}\right)\right] \nonumber\\
	    &= H\left(\hspace{-1pt}X(\ell) | X(1\!:\! \ell\!-\!1\hspace{-1pt}),U_{{0}}\right)- H\left(\hspace{-1pt}X(\ell)|X(1\!\!:\! \ell\!-\!1)\hspace{-1pt},U_{K}\right)\nonumber\\
	    &\leq H(X(\ell)) \leq 1,
	\end{align}
	where we define $U_{0} \!=\! \varnothing$ as a dummy variable and ${U_{K}\!=\!X}$. 
	Therefore, using~\eqref{eq:sum_Q} for each $\ell\in[B]$ we can write
	\begin{align}\label{eq:thm2:1}
	    \sum_{k=1}^{K} \w_{k}\overline{F}_{\widetilde{L}_{k}}(\ell) \cdot Q_{\ell,k} & \leq \left(\max_{k} \w_{k}\overline{F}_{\widetilde{L}_{k}}(\ell)\right) \cdot \sum_{k=1}^K Q_{\ell,k}\nonumber\\
        &\leq \max_{k} \w_{k}\overline{F}_{\widetilde{L}_{k}}(\ell).
	\end{align}
Thus, plugging~\eqref{eq:thm2:1} into~\eqref{eq:thm2_0} we get
\begin{align}\label{eq:f_bnd_w}
    f(\cache,\bd)\sum_{k=1}^{K}
	\omega_{k}\sum_{k=1}^K \!\w_{k}\! \lp 1- \mu_{[k]}\rp & \!\leq\! 
	\max_{k} \sum_{\ell=1}^\nB \max_{k}  \w_{k}\overline{F}_{\widetilde{L}_{k}}(\ell)\nonumber\\
    & \!=\! \sum_{\ell=1}^\nB \max_{k}  \w_{k}\overline{F}_{L_{k}}(\ell), 
\end{align}
where the last equality follows from Lemma~\ref{lm:chnl_ench}-(iv). 
Dividing both sides of~\eqref{eq:f_bnd_w} by $\sum_{k=1}^K \!\w_{k}\! \lp 1- \mu_{[k]}\rp$, we arrive at 
\begin{align}\label{eq:f_idnty_per}
    f(\cache,\bd) \leq \frac{\sum_{\ell=1}^\nB \max_{k}  \w_{k}\overline{F}_{L_{k}}(\ell)}{\sum_{k=1}^K \!\w_{k}\! \lp 1- \mu_{[k]}\rp}.
\end{align}
Now we return to examine the upper bound for an arbitrary weight vector ${\bfw=(\w_1,\ldots,\w_k) \in [0,\infty)^K}$. Let $\pi$ be a permutation that  sorts the vector $\bfw$ in a non-increasing order, i.e., ${\w_{\pi(1)}\!\geq\! \cdots \!\geq\! \w_{\pi(K)}}$. Now, applying the enhancement in~\eqref{eq:chnl_ench}, we arrive at a set of (statistically) degraded channels,
\[
X\leftrightarrow \wdt{Y}_{\pi(K)}\leftrightarrow \wdt{Y}_{\pi(K-1)}\leftrightarrow \cdots 
\leftrightarrow \wdt{Y}_{\pi(2)}
\leftrightarrow \wdt{Y}_{\pi(1)}.
\]
Repeating the argument above, we arrive at~\eqref{eq:f_idnty_per} for a permuted version of the variables, that is,
\begin{align}\label{eq:fstar_v_1}
    f(\cache,\bd) \leq \frac{\sum_{\ell=1}^\nB \max_{k}  \w_{\pi(k)}\overline{F}_{L_{\pi(k)}}(\ell)}{\sum_{k=1}^{K}
	\w_{\pi(k)}\! \lp 1 - \mu_{\pi([k])}\rp}.
\end{align}
By minimizing the right hand side of~\eqref{eq:fstar_v_1} over all non-negative vectors $\bfw$, we get the desired bound, i.e.,
\begin{align}\label{eq:fstar_v_2}
    f(\cache,\bd) & \leq \fs(\cache,\bd)\nonumber\\
    & = \min_{\bfw \geq 0} \frac{\sum_{\ell=1}^\nB \max_{k}  \w_{\pi(k)}\overline{F}_{L_{\pi(k)}}(\ell)}{\sum_{k=1}^{K}
	\w_{\pi(k)}\! \lp 1 - \mu_{\pi([k])}\rp}.
\end{align} 
This completes the proof of the theorem. 
\end{proof}
\subsection{An LP Representation}
The main step in the proof of Proposition~\ref{pro:LP-fbnd} is to define a new weight vector $\bfs$, which allows us to transform the upper-bound in~\eqref{eq:fstar_v_2} into a linear form. 
\begin{proof}[Proof of Proposition~\ref{pro:LP-fbnd}]
Using ${\w_{\pi(1)}\! \geq \cdots \geq \!\w_{\pi(K)} \geq 0}$ and starting from~\eqref{eq:fstar_v_2}, we can write
  \begin{align}\label{eq:f_LP_1}
        \fs(\cache,\bd) & =\min_{\bfw \geq 0} \frac{\sum_{\ell=1}^\nB \max_{k}  \w_{\pi(k)}\overline{F}_{L_{\pi(k)}}(\ell)}{\sum_{k=1}^{K} \w_{\pi(k)}\! \lp 1 - \mu_{\pi([k])}\rp}\nonumber\\
        & = \min_{\pi \in \Pi} \min_{\bfw\in \Omega_\pi} \frac{\sum_{\ell=1}^\nB \max_{k}  \w_{\pi(k)}\overline{F}_{L_{\pi(k)}}(\ell)}{\sum_{k=1}^{K} \w_{\pi(k)}\! \lp 1 - \mu_{\pi([k])}\rp},
    \end{align}
    where $\Omega_\pi:= \left\{\bfw: \w_{\pi(1)} \geq \cdots \geq \w_{\pi(K)}\geq 0 \right\}$.
    Now, we fix some  $\pi\in \Pi$ and focus on the inner minimization in~\eqref{eq:f_LP_1}, i.e., 
    \begin{align}
        \fs_\pi(\cache,\bd) 
        & = \min_{\bfw\in \Omega_\pi} \frac{\sum_{\ell=1}^\nB \max_{k}  \w_{\pi(k)}\overline{F}_{L_{\pi(k)}}(\ell)}{\sum_{k=1}^{K} \w_{\pi(k)}\! \lp 1 - \mu_{\pi([k])}\rp}.
    \end{align}
    We define ${\sigma_k := \w_{\pi(k)} \lp 1 \!-\! \mu_{\pi([k])}\rp\hspace{-1pt}/\sum_{u=1}^K \w_{\pi(u)} \lp 1 \!-\! \mu_{\pi([u])}\rp}$ for every $k \in [K]$. We note that the vector $\bfs:=(\sigma_1,\ldots,\sigma_K)$ satisfies the following conditions:
    \begin{enumerate}[label=\bf{(C\arabic*)}, ref=\bf{(C\arabic*)}]
    \item  \label{cnd1} Since $\w_{\pi(k)} \geq 0$ and $\mu_{\pi([k])}\leq 1$ for every $k\in [K]$, we have $\sigma_k \geq 0$ for every $k\in[K]$;
    \item \label{cnd2} We have
    \begin{align*}
    \sum_{k=1}^{K} \sigma_k = \sum_{k=1}^{K} \frac{\w_{\pi(k)} \lp 1 - \mu_{\pi([k])}\rp}{\sum_{u=1}^K \w_{\pi(u)} \lp 1 - \mu_{\pi([u])}\rp } = 1;
    \end{align*}
    \item \label{cnd3} Using $\w_{\pi(k-1)} \geq \w_{\pi(k)}$, we get
    \begin{align*}
        \frac{\sigma_{k-1}}{1 - \mu_{\pi([k-1]),}} \geq \frac{\sigma_{k}}{1 - \mu_{\pi([k])}} ,
    \end{align*}
    or equivalently,
    \begin{align*}
         \sigma_{k-1}\cdot \lp 1 - \mu_{\pi([k])}\rp \geq \sigma_k\cdot \lp 1 - \mu_{\pi([k-1])}\rp,
    \end{align*}
    for every $k \in \{2,\cdots, K \}$.
\end{enumerate}
    We define $\Sigma_{\pi}$ as the set of all vectors $\bfs$ satisfying three conditions in~\ref{cnd1}-\ref{cnd3}.
    
    Note that for every vector $\bfw\in \Omega_{\pi}$, there is a vector $\bfs \in \Sigma_{\pi}$ and vice versa. Applying this change of variables in ~\eqref{eq:f_LP_1}, we arrive at
    \begin{align}\label{eq:f_LP_2}
        \fs_\pi(\cache,\bd)\! &=\!\!  \min_{\bfw\in \Omega_\pi} \hspace{-2pt}\sum_{\ell=1}^B \max_k \frac{\w_{\pi(k)} \lp 1 \hspace{-1pt}-\hspace{-1pt} \mu_{\pi([k])}\rp}{\sum_{u=1}^K \w_{\pi(u)} \lp 1 \hspace{-1pt}-\hspace{-1pt} \mu_{\pi([u])}\rp}\frac{\ov{F}_{L_{\pi(k)}}(\ell)}{ 1 \hspace{-1pt}-\hspace{-1pt} \mu_{\pi([k])}}
        \nonumber\\
        & =   \min_{\bfs\in \Sigma_{\pi}} \sum_{\ell=1}^B \max_k \sigma_k \frac{\ov{F}_{L_{\pi(k)}}(\ell)}{ 1 \hspace{-1pt}- \hspace{-1pt}\mu_{\pi([k])}}\!.
    \end{align}
   Let us consider each summand in~\eqref{eq:f_LP_2}. For a given $\pi\in\Pi$ and $\bfs\in \Sigma_\pi$, the $\ell$th term in the summation is
\begin{align}\label{eq:op-theta-1}
    \max_k \sigma_k \frac{\ov{F}_{L_{\pi(k)}}(\ell)}{ 1 \!-\! \mu_{\pi([k])}} \!=\! \min \left\{\theta_\ell \!:\! \theta_\ell \!\geq\! \sigma_k \frac{\ov{F}_{L_{\pi(k)}}(\ell)}{1 - \mu_{\pi([k])}}, k\in[K] \right\}.
\end{align}
Let us define
\begin{align}\label{eq:Theta}
   \!\Theta_{\pi}^{\bfs} \!:=\! \left\{\!\bft\!=\!(\theta_1,\dots, \theta_B) \!:\! \theta_\ell \!\geq \!\sigma_k\frac{\ov{F}_{L_{\pi(k)}}(\ell)}{1 - \mu_{\pi([k])}}, k\in[K]\!\right\}\!.
\end{align}
Then,~\eqref{eq:op-theta-1} can be written as
\begin{align}\label{eq:op-theta-1:0}
    \max_k \sigma_k \frac{\ov{F}_{L_{\pi(k)}}(\ell)}{ 1 \!-\! \mu_{\pi([k])}} \!=\! \min_{\bft\in \Theta_{\pi}^{\bfs}} \theta_\ell.
\end{align}

Note that the conditions on each $\theta_\ell$ in~\eqref{eq:Theta} only depend on $\{\ov{F}_{L_{\pi(k)}}(\ell)\}_k$, and hence, for $\ell\neq \ell'$, the conditions on $\theta_\ell$ and $\theta_{\ell'}$ are independent of each other. In other words, $\Theta_{\pi}^{\bfs}$ is an orthant with an offset in $\mathbb{R}^B$. Hence, the minimum of the summation of $\{\theta_{\ell}\}_{\ell=1}^B$ and the summation of the minimum of  $\{\theta_{\ell}\}_{\ell=1}^B$ are equivalent, i.e.,
\begin{align}\label{eq:op-theta-2}
\sum_{\ell=1}^B  \min_{\bft \in \Theta_{\pi}^{\bfs}} \theta_{\ell} 
=
    \min_{\bft \in \Theta_{\pi}^{\bfs}} \sum_{\ell=1}^B \theta_{\ell} . 
\end{align}
Combining~\eqref{eq:op-theta-1:0} and~\eqref{eq:op-theta-2}, we arrive at
\begin{align}\label{eq:op-theta-3}
    \sum_{\ell=1}^B \max_k \sigma_k \frac{\ov{F}_{L_{\pi(k)}}(\ell)}{1 - \mu_{\pi([k])}} = \min_{\bft \in \Theta_{\pi}^{\bfs}} \sum_{\ell=1}^B \theta_{\ell}.
\end{align}
Plugging~\eqref{eq:op-theta-3} into~\eqref{eq:f_LP_2}, we get
\begin{align}\label{eq:f_LP_3}
    \fs_\pi(\cache,\bd) =\min_{\bfs\in \Sigma_{\pi}} \min_{\bft \in \Theta_{\pi}^{\bfs}} \sum_{\ell=1}^B \theta_{\ell}.
\end{align}
Let $\bx := [\bfs,\bft]$. It is important to note that the objective function and the constraints on the optimization problem~\eqref{eq:f_LP_3} are linear in $\bx$. More precisely, using matrices $\bA$ and $\bb$ defined in 
and~\eqref{eq:Matrix-A} in~\eqref{eq:vector-b}, conditions~\ref{cnd1},~\ref{cnd2}, and~\ref{cnd3} on vector $\bfs\in\Sigma_\pi$ can be translated into $-\bx \leq 0$, $\bb \bx =1$, and the lower $K-1$ rows of $\bA_\pi \bx \leq 0$, respectively. Moreover, the constraints on $\bft \in \Theta_{\pi}^{\bfs}$ in~\eqref{eq:Theta} can be expressed as the top $KB$ rows of $\bA_\pi \bx \leq 0$. Therefore, we can rewrite~\eqref{eq:f_LP_3} as
%
 \begin{align*}
     \fs(\cache,\bd) =\ & \min_{\bx \in \mathbb{R}^{K\!+\!B}} \ [\zeros_{K}^T,\ones_{ B }^T] \bx\\
    & \ \ \ \textrm{s.t.\ } \ \bA_{\pi} \bx \leq \zeros,\nonumber\\
    &\phantom{\ \ \ \ \ \textrm{s.t.}} \ \ \bb \bx = 
    1,\nonumber\\ 
    &\phantom{\ \ \ \ \textrm{s.t.}}\ -\bx \leq \zeros\nonumber. 
\end{align*}
This completes the proof of the proposition. 
\end{proof}
\section{Converse Proof of Theorem~\ref{thm:degraded}}\label{sec:proof-degraded-conv}
The converse proof of Theorem~\ref{thm:degraded} is derived directly from the proof of Theorem~\ref{thm:Kuser-deg-BC} where no channel enhancement is required, i.e., $\wdt{L}_k = L_k$ for every $k\in [K]$. 
Since $L_K\geq_{\mathsf{st}}\cdots \geq_{\mathsf{st}} L_1$ the $K$-DTVBC is degraded, we can repeat the steps~\eqref{eq:I_bnd} in through~\eqref{eq:lm:Kuser-deg-BC:simplified} with no further channel enhancement, i.e., $\ov{F}_{\wdt{L}_k}(\ell) = \ov{F}_{L_k}(\ell)$ for every $k\in [K]$ and $\ell\in [B]$. Hence, for the cache placement strategy $\cent$ and its caching tuple $\tcent$ we can write
\begin{align}\label{eq:deg_f_0}
f(\cache,\bd) \cdot \lp 1 - \tcents_{[k]}\rp \leq \sum_{\ell=1}^\nB \overline{F}_{L_{k}}(\ell) Q_{\ell,{k}}, \quad \forall k\in [K],
\end{align}
where
\[{Q_{\ell,{k}} \!=\! H\!\left(X(\ell) \md   X(1\!:\! \ell\!-\!1),U_{{k-1}}\right) \!-\! H\!\left(X(\ell)\md  X(1\!\!:\! \ell\!-\!1),U_{k}\right)}\]
for the Markov chain $U_{{k-1}}\leftrightarrow U_{k}\leftrightarrow X(\ell)$. It is easy to verify that $Q_{\ell,{k}} \geq 0$. Moreover, from~\eqref{eq:sum_Q}, we have
\begin{align}\label{eq:const_degraded}
	    \sum_{k=1}^K Q_{\ell,{k}} \leq 1, \quad \forall \ell\in [B].
\end{align}
From~\eqref{eq:deg_f_0} and~\eqref{eq:const_degraded}, we can write 
\begin{align}
\begin{split}\label{eq:degraded-conv-f}
     & f(\cache,\bd) \leq \max \ \bar{f} \\
     & \textrm{s.t.} \quad \bar{f}\cdot \lp 1 \!-\! \tcents_{[k]}\rp \!\leq\! \sum_{\ell=1}^\nB \overline{F}_{L_{k}}(\ell) Q_{\ell,{k}}, \quad \forall k\in [K],\\
     &\phantom{\textrm{s.t.}\quad} \sum_{k=1}^K Q_{\ell,{k}} \leq 1, \quad \forall\ell\in [B].
\end{split}
\end{align}
Noting the LP problems in~\eqref{eq:fwKUser_degraded} with constraints~\eqref{eq:const_thm2} and~\eqref{eq:degraded-conv-f} are equivalent, we arrive at the claim of Theorem~\ref{thm:degraded}. This completes the proof of the theorem.
\hfill $\square$ 

\section{Conclusion}\label{sec:conc}
In this work, we studied a $K$-user coded-caching problem in a joint source-channel coding framework by providing each user a cache. The transmitter has a certain rate for all files per channel use, and a fraction of the bits/symbols are available in each user's cache. After this, each user requests a file from the database where the transmitter needs to satisfy users' demands over the $K$-DTVBC. The receivers have only access to the channel state information. We characterized the maximum achievable source rate for the $2$-DTVBC and the degraded $K$-DTVBC. Then, we provided an upper bound for the source rate with any caching strategy $\cache$. Finally, we presented an achievable scheme with the LP formulation to show that the upper bound is not a sharp characterization.  
    
Several avenues are left for future research, including closing the gap between the achievable and optimum rates and studying the similar coded-caching problem over the Gaussian fading BC with (un-)coded cache placement schemes.

\appendices

\section*{Appendix: Proof of Lemmas}
In this section, we provide the proofs of lemmas. Note that, in order to avoid repetition,  we provide the proof of Lemma~\ref{lm:2user-deg-BC} after the proof of Lemma~\ref{lm:Kuser-deg-BC}, since most of the techniques used in the latter are also applied in the former. 
\begin{proof}[Proof of Lemma~\ref{lm:2user_f_ach}]
    The first part of the lemma is an immediate consequence of the level-allocation in~\eqref{eq:cent-2:ind-com} and~\eqref{eq:cent-2:ind-sum}. 
    
    Before we prove the other claims of the lemma, consider the LHS of~\eqref{eq:f_1_u_a} and note that the first term in the minimization increases with respect to both $u$ and $\alpha$, while the second term decreases with $u$ and $\alpha$. Hence, the minimum of two terms is maximized when two terms are equal. Thus, we can write
 \begin{align}
     f^\s_1 & = \frac{1}{1\!-\!2\mu}\left[\sum_{i<u^\s} \ov{F}_{L_1}(\el_i) \!+\!\alpha^\s \ov{F}_{L_1}(\el_{u^\s}) \right] \nonumber\\
     & = \frac{1}{1\!-\!\mu}\left[\sum_{i>u^\s} \ov{F}_{L_2}(\el_i) +(1-\alpha^\s) \ov{F}_{L_2}(\el_{u^\s}) \right].\label{eq:f_s_1}
 \end{align}
 Similarly, for~\eqref{eq:f_2_v_b} we get
 \begin{align}
     f^\s_2 & = \frac{1}{1-\mu}\left[\sum_{i<v^\s} \ov{F}_{L_1}(\el_i) \!+\! (1-\beta^\s) \ov{F}_{L_1}(\el_{v^\s}) \right]\nonumber\\
     & = \frac{1}{1\!-\!2\mu}\left[\sum_{i>v^\s} \ov{F}_{L_2}(\el_i) + \beta^\s  \ov{F}_{L_2}(\el_{v^\s}) \right].\label{eq:f_s_2}
 \end{align}
Now, in order to prove the second and the third claims of the lemma, we can distinguish two cases, depending on whether $f_1^\star \leq f_2^\star$ or $f_1^\star \geq f_2^\star$. Let us start with the first case. Since $\mu>0$,  we can write
 \begin{align}\label{eq:u-v:1}
    \!\sum_{i\leq v^\s} \!\ov{F}_{L_1}\!(\el_i)\! & > \frac{1-2\mu}{1\!-\!\mu}\left[\sum_{i\leq v^\s} \ov{F}_{L_1}(\el_i) \right] \nonumber\\
    & \stackrel{\rm{(a)}}{\geq} 
    \frac{1-2\mu}{1\!-\!\mu}\left[\sum_{i<v^\s} \ov{F}_{L_1}(\el_i) + (1-\beta^\s) \ov{F}_{L_1}(\el_{v^\s}) \right]\nonumber \\
     & \stackrel{\rm{(b)}}{\geq} (1-2\mu)f_2^\star\nonumber\\
     & \stackrel{\rm{(c)}}{\geq} (1-2\mu)f_1^\star \nonumber\\
     & \stackrel{\rm{(d)}}{=}\sum_{i<u^\s} \ov{F}_{L_1}(\el_i) + \alpha^\s \ov{F}_{L_1}(\el_{u^\s}) \nonumber\\
     & \stackrel{\rm{(e)}}{\geq}\sum_{i<u^\s} \ov{F}_{L_1}\!(\el_i),
 \end{align}
 where $\rm{(a)}$ holds since $\beta^\star\leq 1$, $\rm{(b)}$ follows from the first equality in~\eqref{eq:f_s_2},  $\rm{(c)}$ 
 is due to assuming $f_1^\star \leq f_2^\star$, $\rm{(d)}$ follows from the first equality in~\eqref{eq:f_s_1}, and $\rm{(e)}$ holds since $\alpha^\star \geq 0$. 
 Then, ~\eqref{eq:u-v:1} implies that $v^\star \geq u^\star$. 
 
 For the second case with $f_1^\star \geq f_2^\star$ and $\mu>0$,  we can write
 \begin{align}\label{eq:u-v:2}
    \!\sum_{i\geq u^\s} \!\!\ov{F}_{L_2}\!(\el_i)\! &> \frac{1-2\mu}{1\!-\!\mu}\left[\sum_{i\geq u^\s} \ov{F}_{L_2}(\el_i) \right] \nonumber\\
    & \stackrel{\rm{(a)}}{\geq} 
    \frac{1-2\mu}{1\!-\!\mu}\left[\sum_{i>u^\s} \ov{F}_{L_2}(\el_i) + (1-\alpha^\s) \ov{F}_{L_2}(\el_{u^\s}) \right]\nonumber \\
     & \stackrel{\rm{(b)}}{=} 
     (1-2\mu)f_1^\star \nonumber\\
     &\stackrel{\rm{(c)}}{\geq} (1-2\mu)f_2^\star \nonumber\\
    &\stackrel{\rm{(d)}}{=} \!\!\sum_{i>v^\s} \!\ov{F}_{L_2}(\el_i) \!+\! \beta^\s \!\ov{F}_{L_2}\!(\el_{u^\s})\nonumber\\
    & \stackrel{\rm{(e)}}{\geq}\!\! \sum_{i>v^\s}\! \ov{F}_{L_2}(\el_i),
 \end{align}
 where $\rm{(a)}$ holds for $\alpha^\star\leq 1$, $\rm{(b)}$ follows from the second equality in~\eqref{eq:f_s_1}, $\rm{(c)}$ is true since we assumed $f_1^\star \geq f_2^\star$, $\rm{(d)}$ follows from the second equality in~\eqref{eq:f_s_2}, and $\rm{(e)}$ holds for $\beta^\star > 0$. From~\eqref{eq:u-v:2}, it can be immediately seen that $v^\star \geq u^\star$, as claimed in part (ii) of the lemma.

 We prove the third claim assuming that ${f^\s_1 \!\leq\! f^\s_2}$. Note that the proof for the other case is very similar. The proof is by contradiction. Assume $g(\ell_{u^\s}) > 1$. Then, since the levels are sorted with respect to $g(\cdot)$, we have ${1<g(\ell_{u^\s}) \leq \cdots \leq g(\ell_B)}$, which implies ${\ov{F}_{L_2}(\ell_i) > \ov{F}_{L_1}(\ell_i)}$ for every $i>u^\s$. Thus, we have 
 \begin{align}
     (1-&\alpha^\s) \ov{F}_{L_2}(\el_{u^\s}) +\sum_{u^\s<i<v^\s} \ov{F}_{L_2}(\el_i)  
     +\beta^\s \ov{F}_{L_2}(\el_{v^\s})+ (1-\beta^\s) \ov{F}_{L_2}(\el_{v^\s}) 
     +\sum_{i>v^\s} \ov{F}_{L_2}(\el_i) 
     \nonumber\\
     & =(1-\alpha^\s) \ov{F}_{L_2}(\el_{u^\s})  \!+\!\sum_{i>u^\s}  \ov{F}_{L_2}(\el_i) \nonumber\\
     &\stackrel{\rm{(a)}}{=}   (1-\mu)f_1^\star  \nonumber\\ 
     & \stackrel{\rm{(b)}}{\leq} (1-\mu) f_2^\star\nonumber\\
     &\stackrel{\rm{(c)}}{=}  \sum_{i<v^\s} \ov{F}_{L_1}(\el_i) +(1-\beta^\s) \ov{F}_{L_1}(\el_{v^\s})\nonumber\\
     & = \sum_{i<u^\s} \ov{F}_{L_1}(\el_i) \!+\! \alpha^\s \ov{F}_{L_1}(\el_{u^\s}) + (1\!-\!\alpha^\s) \ov{F}_{L_1}(\el_{u^\s}) \!+\! \sum_{u^\s<i<v^\s} \!\ov{F}_{L_1}(\el_i) + (1-\beta^\s) \ov{F}_{L_1}(\el_{v^\s}),
     \label{eq:f_1_inq2}
 \end{align}
where $\rm{(a)}$ follows from the second equality in~\eqref{eq:f_s_1},  $\rm{(b)}$ holds as $f_1^\s\leq f^\s_2$,  and step $\rm{(c)}$ follows from the first equality in~\eqref{eq:f_s_2}. 
Subtracting the RHS of~\eqref{eq:f_1_inq2} from its LHS, we arrive at 
\begin{align}\label{eq:pr:lm:2:3}
0&\leq \Bigg[
\sum_{i<u^\s} \ov{F}_{L_1}(\el_i) \!+\! \alpha^\s \ov{F}_{L_1}(\el_{u^\s}) + (1\!-\!\alpha^\s) \ov{F}_{L_1}(\el_{u^\s}) +\! \sum_{u^\s<i<v^\s} \!\ov{F}_{L_1}(\el_i) + (1-\beta^\s) \ov{F}_{L_1}(\el_{v^\s})
\Bigg]\nonumber\\
&\phantom{\leq} - \Bigg[ (1-\alpha^\s) \ov{F}_{L_2}(\el_{u^\s}) +\sum_{u^\s<i<v^\s} \ov{F}_{L_2}(\el_i) +\beta^\s \ov{F}_{L_2}(\el_{v^\s})+ (1-\beta^\s) \ov{F}_{L_2}(\el_{v^\s}) 
     +\sum_{i>v^\s} \ov{F}_{L_2}(\el_i) \Bigg]\nonumber\\
     &=\Bigg[\sum_{i<u^\s} \ov{F}_{L_1}(\el_i) \!+\! \alpha^\s \ov{F}_{L_1}(\el_{u^\s})\Bigg] -\Bigg[\sum_{i>v^\s} \ov{F}_{L_2}(\el_i) +\beta^\star \ov{F}_{L_2}(\el_{v^\star})\Bigg]\nonumber\\
     &\phantom{=} + \Bigg[(1-\alpha^\star) \Big(\ov{F}_{L_1}(\el_{u^\star}) - \ov{F}_{L_2}(\el_{u^\star})\Big) \Bigg]+ \Bigg[(1-\beta^\star) \Big(\ov{F}_{L_1}(\el_{v^\star}) - \ov{F}_{L_2}(\el_{v^\star})\Big) \Bigg]\nonumber\\
     &\phantom{=} + \Bigg[
     \sum_{u^\s<i<v^\s} \Big(\ov{F}_{L_1}(\el_i)  
     -   \ov{F}_{L_2}(\el_i)\Big) \Bigg]\nonumber\\
     &\stackrel{\rm{(a)}}{<} \Bigg[\sum_{i<u^\s} \ov{F}_{L_1}(\el_i) \!+\! \alpha^\s \ov{F}_{L_1}(\el_{u^\s})\Bigg] -\Bigg[\sum_{i>v^\s} \ov{F}_{L_2}(\el_i) +\beta^\star \ov{F}_{L_2}(\el_{v^\star})\Bigg]\nonumber\\
     &\stackrel{\rm{(b)}}{=} (1-2\mu) f_1^\star - (1-2\mu) f_2^\star = (1-2\mu)(f_1^\star - f_2^\star),
\end{align}
where $\rm{(a)}$ holds since $\ov{F}_{L_1}(\el_i) <  \ov{F}_{L_2}(\el_i)$ for $i>u^\star$, and $\rm{(b)}$ follows from the second equalities in~\eqref{eq:f_s_1} and~\eqref{eq:f_s_2}. Then, \eqref{eq:pr:lm:2:3} implies that $f_1^\star >f_2^\star$, which is in contradiction with the assumption that $f_1^\star \leq f_2^\star$. Hence, we can conclude that $g(\el_{u^\s})\leq 1$. This completes the proof of part (iii). 

Finally, we can prove part (iv) of the lemma. 
First, assume that ${f^\s_1\leq f^\s_2}$. 
From~\eqref{eq:f_s_1}, we can write
 \begin{align*}
          \frac{1}{g(\ell_{u^\s})}(1\hspace{-1pt}-\hspace{-1pt}\mu)f_1^\s   \!&=\! \frac{1}{g(\ell_{u^\s})}\!\left[\sum_{i>u^\s} \hspace{-1pt}\ov{F}_{L_2}(\el_i) \!+\!(1\!-\!\alpha^\s) \ov{F}_{L_2}(\el_{u^\s})\!\right]\nonumber\\
          &\stackrel{\rm{(a)}}{=} \frac{1}{g(\ell_{u^\s})} \sum_{i>u^\s} \ov{F}_{L_2}(\el_i) \!+\!(1\hspace{-1pt}-\hspace{-1pt}\alpha^\s) \ov{F}_{L_1}(\el_{u^\s}),
\end{align*}
where in $\rm{(a)}$ follow from ${g(\ell_{u^\star}) = \ov{F}_{L_2}(\el_{u^\s}) / \ov{F}_{L_1}(\el_{u^\s})}$. Moreover, from~\eqref{eq:f_s_1} we have 
\begin{align*}
        (1-2\mu) f_1^\s \!=\! \sum_{i<u^\s} \ov{F}_{L_1}(\el_i) +\alpha^\s \ov{F}_{L_1}(\el_{u^\s}). 
     \end{align*}
     Combining these two equations, we arrive at
     \begin{align}\label{eq:pr:lm:2:4}
         f_1^\s & = \frac{\sum_{i\leq u^\s} \ov{F}_{L_1}(\el_i) + \frac{1}{g(\ell_{u^\s})}\sum_{i>u^\s} \ov{F}_{L_2}(\el_i)}{ (1-2\mu)+ \frac{1}{g(\ell_{u^\s})}(1-\mu)}\nonumber\\
         & \stackrel{\rm{(a)}}{=} \frac{R_1(g(\ell_{u^\s})) + \frac{1}{g(\ell_{u^\s})}R_2(g(\ell_{u^\s}))}{ (1-2\mu)+ \frac{1}{g(\ell_{u^\s})}(1-\mu)}\nonumber\\
         &= \frac{R_1(\omega) + \frac{1}{\omega)}R_2(\omega)}{ (1-2\mu)+ \frac{1}{\omega}(1-\mu)}\Bigg|_{\omega=g(\ell_{u^\s})}\nonumber\\
         &\stackrel{\rm{(b)}}{\geq}   \min_{\omega\leq 1}\frac{R_1(\omega) + \frac{1}{\omega)}R_2(\omega)}{ (1-2\mu)+ \frac{1}{\omega}(1-\mu)},
     \end{align}
     where in~$\rm{(a)}$ we have $R_1(g(\ell_{u^\s}))$ and $R_1(g(\ell_{u^\s}))$ as given in~\eqref{eq:R1-R2}. Moreover, 
     \begin{align*}
         \cL_1(g(\ell_{u^\s})) &=
         \{\ell: g(\ell_{u^\s}) \ov{F}_{L_1}(\el)  \geq \ov{F}_{L_2}(\el)  \} \nonumber\\
         &=\{\ell: g(\ell_{u^\s})  \geq g(\el)  \} 
         =\{i: i\leq u^\s\},
     \end{align*}
         and ${\cL_2(g(\ell_{u^\s})) = \{i: i>u^\s\}}$ as indicated in the statement of Theorem~\ref{thm:2-User}. Also, note that from part~(iii), we have $g(\el_{u^\s})\leq 1$, which justifies the inequality in~$\rm{(b)}$. 
     
     Now, we consider the case that $f^\s_2\leq f^\s_1$. From~\eqref{eq:f_s_2}, we can write
     \begin{align*}
         g(\ell_{v^\s})(1-\mu)f_2^\s   \!&=\! g(\ell_{v^\s})\!\left[\sum_{i<v^\s} \ov{F}_{L_1}(\el_i) \!+\!(1\hspace{-1pt}-\hspace{-1pt}\beta^\s) \ov{F}_{L_1}(\el_{v^\s})\right]\hspace{-1pt}.
     \end{align*}
     Further, we get
     \begin{align*}
        (1-2\mu) f_2^\s & = \sum_{i>v^\s} \ov{F}_{L_2}(\el_i) +\beta^\s \ov{F}_{L_2}(\el_{v^\s})\\
        &\stackrel{\rm{(a)}}{=} 
          \sum_{i>v^\s} \ov{F}_{L_2}(\el_i) + g(\ell_{v^\s}) \beta^\s \ov{F}_{L_1}(\el_{v^\s}),
     \end{align*}
     where $\rm{(a)}$  follows from ${g(\ell_{v^\star}) \!=\! \ov{F}_{L_2}(\el_{v^\s}) / \ov{F}_{L_1}(\el_{v^\s})}$. 
     Combining these two equations, we have
     \begin{align}\label{eq:pr:lm:2:4_2}
         f_2^\s & = \frac{g(\ell_{v^\s})\sum_{i\leq v^\s} \ov{F}_{L_1}(\el_i) + \sum_{i>v^\s} \ov{F}_{L_2}(\el_i)}{g(\ell_{v^\s})(1-\mu)+ (1-2\mu)}\nonumber\\
         & = \frac{g(\ell_{v^\s}) R_1(g(\ell_{v^\s})) +  R_2(g(\ell_{v^\s}))}{ g(\ell_{v^\s})(1-\mu)+ (1-2\mu)}\nonumber\\
         &= \frac{\omega R_1(x\omega) +  R_2(\omega)}{ \omega(1-\mu)+ (1-2\mu)}\Bigg|_{\omega = g(\ell_{v^\s})}\nonumber\\
         &\geq \min_{0\leq \omega \leq 1} \frac{\omega R_1(x\omega) +  R_2(\omega)}{ \omega(1-\mu)+ (1-2\mu)}.
     \end{align}
Here, $R_1(\omega)$ and $R_2(\omega)$ are defined as in~\eqref{eq:R1-R2}. Moreover, we have 
     \begin{align*}
         \cL_1(g(\ell_{v^\s})) &=
         \{\ell: g(\ell_{v^\s}) \ov{F}_{L_1}(\el)  \geq \ov{F}_{L_2}(\el)  \} \nonumber\\
         &=\{\ell: g(\ell_{v^\s})  \geq g(\el)  \} 
         =\{i: i\leq v^\s\},
     \end{align*}
     and $\cL_1(g(\ell_{v^\s})) = \{i: i> v^\s\}$. 
     It is worth noting that, the last inequality in~\eqref{eq:pr:lm:2:4_2} holds since  $g(\ell_{v^\s})\geq 1$, as shown in part~(iii) of the lemma. 

     Combining~\eqref{eq:pr:lm:2:4} and~\eqref{eq:pr:lm:2:4_2}, we arrive at ${\min(\fs_1, \fs_2)\geq \fs}$, for $\fs$ defined in \eqref{eq:f-2User}. This completes the proof of part~(iv). 
\end{proof}
\begin{proof}[Proof of Lemma~\ref{lm:chnl_ench}]
In order to prove (i), we show that  $\ov{F}_{\wdt{L}_{k}}(\ell)=1$ implies $\ov{F}_{\wdt{L}_{k+1}}(\ell)=1$.  From~\eqref{eq:chnl_ench}, we have
\begin{align*}
    \overline{F}_{\widetilde{L}_{k+1}}(\ell)&=\min\left[1,\max\left(\overline{F}_{L_{k+1}}(\ell),\frac{\w_k}{\omega_{k+1}}\overline{F}_{\widetilde{L}_{{k}}}(\ell)\right)\right]\\
    & =
    \min\left[1,\max\left(\overline{F}_{L_{k+1}}(\ell),\frac{\omega_{{k}}}{\omega_{k+1}}\right)\right]\\
    & \stackrel{\rm (a)}{=} \min\left[1,\frac{\omega_{{k}}}{\omega_{k+1}}\right] \stackrel{\rm (b)}{=}1,
\end{align*}
where $\rm(a)$ and $\rm (b)$ hold because $\ov{F}_{L_{k+1}}(\ell)\leq 1\leq \frac{\omega_{k}}{\omega_{k+1}}$, for a non-increasing sequence $\omega_1\geq \omega_2 \geq \cdots \geq \omega_K$. 
This implies that $\ov{F}_{\wdt{L}_{u}}(\ell)=1$ for every $u\geq k$. 

Next, we prove part~(ii). 
First, assume that 
 ${\!\w_{k} \ov{F}_{L_{k}}\!(\ell) <  \w_{k-1}\ov{F}_{\wdt{L}_{k-1}}\!(\ell)}$. 
 Then, we can write
\begin{align*}
    \overline{F}_{\widetilde{L}_{k}}(\ell)  &= \min\left[1,\max\left(\overline{F}_{L_{k}}(\ell),\frac{\omega_{{k-1}}}{\omega_{k}}\overline{F}_{\widetilde{L}_{{k-1}}}(\ell)\right)\right]\\
   &\leq \hspace{-1pt}\max\left(\overline{F}_{L_{k}}(\ell),\frac{\omega_{{k-1}}}{\omega_{k}}\overline{F}_{\widetilde{L}_{{k-1}}}(\ell)\right) \hspace{-1pt}=\hspace{-1pt}  \frac{\w_{k-1}}{\w_{k}}\ov{F}_{\wdt{L}_{k-1}}(\ell),
\end{align*}
which implies $\w_k \overline{F}_{\widetilde{L}_{k}}(\ell) \leq  \w_{k-1}\overline{F}_{\widetilde{L}_{k-1}}(\ell)$, which is in contradiction with the assumption of part~(ii). Hence, we have ${\!\w_{k} \ov{F}_{L_{k}}\!(\ell)\!\geq\!  \w_{k-1}\ov{F}_{\wdt{L}_{k-1}}\!(\ell)}$. This together with~\eqref{eq:chnl_ench} leads to
\begin{align*}
    \overline{F}_{\widetilde{L}_{k}}(\ell)&=\min\left[1,\max\left(\overline{F}_{L_{k}}(\ell),\frac{\omega_{{k-1}}}{\w_k}\overline{F}_{\widetilde{L}_{{k-1}}}(\ell)\right)\right]\\
    & = \min \left [1, \ov{F}_{L_{k}}(\ell)\right] =  \ov{F}_{L_{k}}(\ell),
\end{align*}
where the last equality follows from $\ov{F}_{L_{k}}(\ell)\leq 1$. This shows the first equality in part~(ii). 

Then, assume $\ov{F}_{\wdt{L}_{u}}(\ell) \!=\! 1$ for some $u<k$.   
From part~(i), we get
\begin{align*}
    \ov{F}_{\wdt{L}_u}(\ell)  =\cdots= \ov{F}_{\wdt{L}_{k-1}}(\ell) =\ov{F}_{\wdt{L}_{k}}(\ell)=1.
\end{align*} 
This together with the fact that $\w_{k-1}\geq \w_{k}$ implies 
${\w_{k-1}\ov{F}_{\wdt{L}_{k-1}}(\ell)\geq \w_{k}\ov{F}_{\wdt{L}_{k}}(\ell)}$ which contradicts with the assumption of part~(ii).  Hence, we have $\overline{F}_{\widetilde{L}_{u}}(\ell)<1$ for every $u<k$. Using this fact, we get 
\begin{align*}
    \ov{F}_{\wdt{L}_{u}}(\ell) &= \min\left[1,\max\left(\overline{F}_{L_{u}}(\ell),\frac{\omega_{{u-1}}}{\omega_{u}}\overline{F}_{\widetilde{L}_{{u-1}}}(\ell)\right)\right]\nonumber\\ 
    & = \max \left(\ov{F}_{L_{u}}(\ell), \frac{\w_{u-1}}{\w_{u}}\ov{F}_{\wdt{L}_{u-1}}(\ell)\right)\nonumber\\
    &\geq \frac{\w_{u-1}}{\w_{u}}\ov{F}_{\wdt{L}_{u-1}}(\ell),
\end{align*}
which results in $\w_{u}\ov{F}_{\wdt{L}_{u}}(\ell) \geq \w_{u-1}\ov{F}_{\wdt{L}_{u-1}}(\ell)$ for every $u<k$, or equivalently
\begin{align*}
   \w_{k}\ov{F}_{\wdt{L}_{k}}\!(\ell) \! > \!\w_{k-1}\ov{F}_{\wdt{L}_{k-1}}\!(\ell) \!\geq\! \w_{k-2}\ov{F}_{\wdt{L}_{k-2}}\!(\ell)\!\geq\! \cdots \!\geq\! \w_{1}\ov{F}_{\wdt{L}_{1}}\!(\ell).
\end{align*}
This completes the proof of part~(ii).

In order to prove part~(iii),  we first note that 
\begin{align}\label{eq:pr:lm:3:3}
    \frac{\w_{k-1}}{\w_{k}}\ov{F}_{\wdt{L}_{k-1}}(\ell) &\stackrel{\rm{(a)}}{>} 
  \ov{F}_{\wdt{L}_{k}}(\ell)\nonumber\\
  &= \min\left[1,\max \left(\ov{F}_{L_{k}}(\ell), \frac{\w_{k-1}}{\w_{k}}\ov{F}_{\wdt{L}_{k-1}}(\ell)\right)\right]\nonumber\\
  & \stackrel{\rm{(b)}}{=}\min\left[1, \frac{\w_{k-1}}{\w_{k}}\ov{F}_{\wdt{L}_{k-1}}(\ell)\right],
\end{align}
where both~$\rm{(a)}$ and~$\rm{(b)}$ follow from the assumption of part~(iii). Then, \eqref{eq:pr:lm:3:3} implies ${\ov{F}_{\wdt{L}_{k}}(\ell) = 1}$. 
This, from part~(i) of the lemma, we get
\begin{align*}
    \ov{F}_{\wdt{L}_{u}}(\ell)=1, \quad u \geq k.
\end{align*}
This along with $\w_{1}\geq \cdots \geq \w_K$ leads to
 \begin{align*}
	  \w_K \ov{F}_{\wdt{L}_K}(\ell) \leq \cdots \leq \w_k \ov{F}_{\wdt{L}_k}(\ell) < \w_{k-1} \ov{F}_{\wdt{L}_{k-1}}(\ell),
\end{align*}
which is the claim of part~(iii). 

Next, we prove part~(iv) of the lemma. Fix some ${\el\in [B]}$, and define $s^\s:=\max \left\{j:  \w_{j}\ov{F}_{\wdt{L}_{j}}(\ell) \!>\! \w_{j-1}\ov{F}_{\wdt{L}_{j-1}}(\ell)\right\}$.  In the following, we first show that
\begin{align}\label{eq:lm_iv_000}
\w_{s^\s}\ov{F}_{\wdt{L}_{s^\s}}(\ell)= \max_k \w_{k}\ov{F}_{\wdt{L}_{k}}(\ell).
\end{align}
The definition of $s^\s$ implies ${\w_{u}\ov{F}_{\wdt{L}_{u}}(\ell)\leq \w_{u+1}\ov{F}_{\wdt{L}_{u+1}}(\ell)}$, for every $u> s^\s$, leading to 
\begin{align}\label{eq:lm_iv_00}
    \w_{s^\s}\ov{F}_{\wdt{L}_{s^\s}}\!(\ell) \geq \w_{s^\s+1} \ov{F}_{\wdt{L}_{s^\s+1}}(\ell) \geq \cdots \geq  \w_{K}\ov{F}_{\wdt{L}_{K}}(\ell).
\end{align}
Moreover, Since $\w_{s^\s} \ov{F}_{\wdt{L}_{s^\s}}(\ell) > \w_{s^\s-1}\ov{F}_{\wdt{L}_{s^\s-1}}(\ell)$, from part~(ii) of the lemma we have
\begin{align}\label{eq:lm_iv_0}
    \w_1\ov{F}_{\wdt{L}_{1}}\!(\ell) \!\leq \!\cdots\! \leq\! \w_{s^\s-1} \ov{F}_{\wdt{L}_{s^\s-1}}\!(\ell) \!<\! \w_{s^\s}\ov{F}_{\wdt{L}_{s^\s}}\!\!(\ell) \!=\! \w_{s^\s} \ov{F}_{L_{s^\s}}\!(\ell).
\end{align}
Combining~\eqref{eq:lm_iv_00} and~\eqref{eq:lm_iv_0}, we can conclude~\eqref{eq:lm_iv_000}. 
Furthermore, we have
\begin{align}\label{eq:lm_iv_1}
    \max_k \w_{k}\ov{F}_{ {L}_{k}}(\ell) &\geq \w_{s^\s} \ov{F}_{L_{s^\s}}\!(\ell) \nonumber\\
    &\stackrel{\rm{(a)}}{=}  \w_{s^\s}\ov{F}_{\wdt{L}_{s^\s}}\!\!(\ell) \nonumber\\
    &= \max_k \w_{k}\ov{F}_{ \wdt{L}_{k}}(\ell) \nonumber\\
    &
    \stackrel{\rm{(b)}}{\geq} 
    \max_k \w_{k}\ov{F}_{ {L}_{k}}(\ell). 
\end{align}
Here, $\rm{(a)}$ follows from~\eqref{eq:lm_iv_0}, and $\rm{(b)}$ is due to the fact that $\w_{k}\ov{F}_{\wdt{L}_{k}}(\ell)\geq \w_{k}\ov{F}_{{L}_{k}}(\ell)$ for every $k\in [K]$. Lastly, \eqref{eq:lm_iv_1} concludes the proof of part~(iv).
\end{proof}

\begin{proof}[Proof of Corollary~\ref{cor:enh_chnl}]
    To prove (i), from the definition of $k^\s$, we get
    \begin{align*}
        \w_{k^\s} \overline{F}_{\wdt{L}_{k^\s}}(\ell)  > \w_{k^\s-1}\ov{F}_{\wdt{L}_{k^\s-1}}(\ell).
    \end{align*}
    This together with Lemma~\ref{lm:chnl_ench}-(ii) for $k=k^\s$ arrives us at
    \begin{align*}
        \w_{k^\s-1}\ov{F}_{\wdt{L}_{k^\s-1}}(\ell) \geq \cdots \geq \w_{1}\ov{F}_{\wdt{L}_{1}}(\ell).
    \end{align*}
     The part (ii) can be directly derived from the definitions of $k^\s$ and $u^\s$. 
     
     In order to prove part~(iii), from the definition of $u^\s$, we have
     \begin{align*}
         \w_{u^\s\!} \ov{F}_{\wdt{L}_{u^\s\!}}(\ell)> \w_{u^\s+1} \ov{F}_{\wdt{L}_{u^\s+1}}(\ell).
     \end{align*}
     This combined with Lemma~\ref{lm:chnl_ench}-(iii) for $k=u^\s+1$ leads us to
     \begin{align*}
         \w_{u^\s\!+1} \ov{F}_{\wdt{L}_{u^\s\!+1}}(\ell)\geq \cdots \geq 
        \w_K \ov{F}_{\wdt{L}_K}(\ell),
     \end{align*}
     which completes the proof of the corollary   
\end{proof}

\begin{proof}[Proof of Lemma~\ref{lm:I-mu}]
We first prove that
\[I\lp W^{(n)}_i;C^{(n)}_\cS \rp = H \lp C^{(n)}_{\cS,i} \rp.\]
To this end, we show $I\lp W^{(n)}_i;C^{(n)}_\cS \rp \!\geq\! H \lp C^{(n)}_{\cS,i} \rp$ and $I\lp W^{(n)}_i;C^{(n)}_\cS \rp \!\leq\! H \lp C^{(n)}_{\cS,i} \rp$. For the first inequality, we can write
\begin{align}\label{eq:ind-0}
    I\lp W^{(n)}_i;C^{(n)}_\cS \rp & = H\lp C^{(n)}_\cS \rp - H\lp C^{(n)}_\cS \md W^{(n)}_i \rp \nonumber\\
    & \geq H\lp C^{(n)}_\cS \rp\nonumber\\
    & \!=\! H\lp C^{(n)}_{\cS,i}, C^{(n)}_{\cS,[N]\setminus\{i\}} \rp \!\geq\! H\!\lp C^{(n)}_{\cS,i} \rp,
\end{align}
where $C^{(n)}_{\cS,i}:= \lp C^{(n)}_{k,i} \rp_{k\in \cS}$ and 
\[C^{(n)}_{\cS,[N]\setminus\{i\}}:=\lp C^{(n)}_{k,1},\ldots, C^{(n)}_{k,i-1},C^{(n)}_{k,i+1},\ldots,C^{(n)}_{k,N} \rp_{k\in \cS}.\] 

On the other hand, we have
\begin{align}\label{eq:ind-1}
    & I\lp W^{(n)}_i;C^{(n)}_\cS \rp\nonumber\\
    & = I\lp W^{(n)}_i; C^{(n)}_{\cS,i}, C^{(n)}_{\cS,[N]\setminus\{i\}} \rp \nonumber\\
    & = I\lp W^{(n)}_i; C^{(n)}_{\cS,[N]\setminus\{i\}} \rp + I\lp W^{(n)}_i; C^{(n)}_{\cS,i} \md C^{(n)}_{\cS,[N]\setminus\{i\}} \rp \nonumber\\
    & \leq I\lp W^{(n)}_i; W^{(n)}_{[N]\setminus \{i\} } \rp  +  I\lp W^{(n)}_i; C^{(n)}_{\cS,i} \md C^{(n)}_{\cS,[N]\setminus\{i\}} \rp\nonumber\\
    & \stackrel{\rm (a)}{=} I\lp W^{(n)}_i; C^{(n)}_{\cS,i} \md C^{(n)}_{\cS,[N]\setminus\{i\}} \rp \nonumber\\
    & = H \lp C^{(n)}_{\cS,i} \md C^{(n)}_{\cS,[N]\setminus\{i\}} \rp  - H \lp C^{(n)}_{\cS,i} \md C^{(n)}_{\cS,[N]\setminus\{i\}}, W^{(n)}_i \rp \nonumber\\
    & = H \lp C^{(n)}_{\cS,i} \md C^{(n)}_{\cS,[N]\setminus\{i\}} \rp  \leq H \lp C^{(n)}_{\cS,i} \rp,
\end{align}
where $\rm{(a)}$ follows since files $W^{(n)}_1,\ldots, W^{(n)}_N$ are mutually independent. Hence, using~\eqref{eq:ind-0},~\eqref{eq:ind-1}, and~\eqref{eq:H-CS}, we arrive at $I\lp W^{(n)}_i;C^{(n)}_\cS \rp = H \lp C^{(n)}_{\cS,i} \rp \leq n\mu_{\cS}f$. This completes the proof of the lemma. 
\end{proof}
\newpage
\begin{proof}[Proof of Lemma~\ref{lm:Kuser-deg-BC}]
	 Let $\boldsymbol{d}=(d_1,\ldots,d_K)$ be the demand vector. First, note that since user $k$ is capable of decoding its requested file $W^{(n)}_{d_k}$ from its received signal and cache content $C^{(n)}_k$, there should exist some family of caching strategies, encoding, and decoding functions with block length $n$ and decoding error probability $\epsilon_n$ where $\epsilon_n\rightarrow 0$ as $n\rightarrow\infty$. From Fano's inequality, we have 
	\begin{align*}
	    H \lp W^{(n)}_{d_k}\md Y^n_k,C^{(n)}_k \rp\leq n\epsilon_n,\quad k\in[K].
	\end{align*}
	Then, we can write
	\begin{align}
    	 nf(\cache,\bd)-n\epsilon_n
            & \leq H\lp W^{(n)}_{d_1}\rp -n\epsilon_n\nonumber\\ &\leq I\lp W^{(n)}_{d_1};Y^n_1,C^{(n)}_1\rp \nonumber\\
    	&= I\lp W^{(n)}_{d_1};Y^n_1\md C^{(n)}_1\rp+I\lp W^{(n)}_{d_1};C^{(n)}_1\rp \nonumber\\
    	&=\sum_{i=1}^{n}I\lp W^{(n)}_{d_1};Y_{1,i}\md Y^{i-1}_1,C^{(n)}_1\rp+I\lp W^{(n)}_{d_1};C^{(n)}_1\rp\nonumber\\
    	&\stackrel{\rm (a)}{\leq} \sum_{i=1}^{n}I\lp W^{(n)}_{d_1},Y^{i-1}_1;Y_{1,i}\md C^{(n)}_1\rp+I\lp W^{(n)}_{d_1};C^{(n)}_1 \rp \nonumber\\
    	&\stackrel{\rm (b)}{=} \!\sum_{i=1}^n \!I\!\lp W^{(n)}_{d_1},Y^{Q-1}_1;Y_{1,Q}\md C^{(n)}_1,Q=i\rp + I\lp W^{(n)}_{d_1};C^{(n)}_1\rp \nonumber\\
    	&= n \sum_{i=1}^n I\lp W^{(n)}_{d_1},Y^{Q-1}_1;Y_{1,Q}\md C^{(n)}_1,Q=i\rp\mathbb{P}(Q=i) +I\lp W^{(n)}_{d_1};C^{(n)}_1\rp\nonumber\\
    	& = n \sum_{i=1}^n I\lp W^{(n)}_{d_1},Y^{Q-1}_1;Y_{1,Q}\md C^{(n)}_1, Q\rp + I\lp W^{(n)}_{d_1};C^{(n)}_1\rp\nonumber\\
    	&
    	\leq nI\lp W^{(n)}_{d_1},C^{(n)}_1, Y^{Q-1}_1, Q;Y_{1,Q}\rp + I\lp W^{(n)}_{d_1};C^{(n)}_1\rp \nonumber\\
    	&
    	\stackrel{\rm (c)}{=} nI(U_1;Y_{1,Q})+I\lp W^{(n)}_{d_1};C^{(n)}_1\rp\nonumber\\
    	& = nI(U_1;Y_1)+I\lp W^{(n)}_{d_1};C^{(n)}_1\rp \nonumber\\ 
    	& \stackrel{\rm (d)}{\leq} nI(U_1;Y_1)+ n\mu_{\{1\}}f(\cache,\bd),\label{eq:nf1_K}
    	\end{align}
    where $\rm{(a)}$ holds since 
    \begin{align*}
       & I\lp W^{(n)}_{d_1},Y^{i-1}_1;Y_{1,i}\md  C^{(n)}_1\rp \\
       & = I\lp Y^{i-1}_1;Y_{1,i}\md  C^{(n)}_1\rp + I\lp W^{(n)}_{d_1};Y_{1,i}\md  Y^{i-1}_1, C^{(n)}_1\rp, 
    \end{align*}
    in $\rm{(b)}$
    $Q$ is a random variable independent of all other random variables which are uniformly distributed over $[n]$, in $\rm{(c)}$ we define ${U_1:= \left(W^{(n)}_{d_1}, C^{(n)}_1, Y^{Q-1}_1, Q\right)}$, and in $\rm{(d)}$ we used~\eqref{eq:I-f}. This implies inequality in~\eqref{eq:lm:f_upper_1}.
    
    We define the subset of indices $d_{[k]}=\{d_1,\ldots,d_k\}$ for every $k\in[K]$.
    	
    Similarly, for $k\in[2:K-1]$, we have
    \begin{align}
    	nf(\cache,\bd)-n\epsilon_n
        &=H\lp W^{(n)}_{d_k} \rp-n\epsilon_n \nonumber\\
    	& \leq I\lp W^{(n)}_{d_k};Y^n_k,C^{(n)}_k\rp \nonumber\\
    	&= I\lp W^{(n)}_{d_k};Y^n_k\md C^{(n)}_k\rp +I\lp W^{(n)}_{d_k};C^{(n)}_k \rp\nonumber\\
    	& 
    	\leq I\lp W^{(n)}_{d_k};Y^n_k,{W^{(n)}_{d_{[k-1]}}},C^{(n)}_{[k-1]}\md C^{(n)}_k \rp \nonumber\\
         &\phantom{\leq} +I\lp W^{(n)}_{d_k};C^{(n)}_k\rp \nonumber\\
    	&
        = I\lp W^{(n)}_{d_k};C^{(n)}_{[k-1]}\md  C^{(n)}_k \rp + I\lp W^{(n)}_{d_k};{W^{(n)}_{d_{[k-1]}}}\md C^{(n)}_{[k]} \rp+I\lp W^{(n)}_{d_k};Y^n_k\md  {W^{(n)}_{d_{[k-1]}}}, C^{(n)}_{[k]}\rp\nonumber\\
        &\phantom{\leq} +I\lp W^{(n)}_{d_k};C^{(n)}_k\rp\nonumber\\
    	&\stackrel{\rm (a)}{=} 
    	I\lp W^{(n)}_{d_k};Y^n_k\md  W^{(n)}_{d_{[k-1]}}, C^{(n)}_{[k]}\rp+I\lp W^{(n)}_{d_k};C^{(n)}_{[k]}\rp\nonumber\\
    	&= 
    	\sum_{i=1}^{n}I\lp W^{(n)}_{d_k};Y_{k,i}\md  {W^{(n)}_{d_{[k-1]}}}, C^{(n)}_{[k]} ,Y_k^{i-1}\rp + I\lp W^{(n)}_{d_k};C^{(n)}_{[k]}\rp \nonumber\\
    	&\leq 
    	\sum_{i=1}^{n}I\lp W^{(n)}_{d_k};Y_{k,i}, Y_{[k-1]}^{i-1}\md  W^{(n)}_{d_{[k-1]}}, C^{(n)}_{[k]} ,Y_k^{i-1}\rp\hspace{-2pt} +\hspace{-2pt}I\lp W^{(n)}_{d_k};C^{(n)}_{[k]}\rp\nonumber\\
    	&\stackrel{\rm (b)}{=}\sum_{i=1}^{n}I\lp W^{(n)}_{d_k};Y_{k,i}\md  \{W^{(n)}_{d_{[k-1]}}, C^{(n)}_{[k]},\!Y_{[k]}^{i-1}\rp\hspace{-2pt}+\hspace{-2pt}I\lp W^{(n)}_{d_k};C^{(n)}_{[k]}\rp\nonumber\\
    	&\leq   
    	\sum_{i=1}^{n}\!I\lp\! W^{(n)}_{d_k},\!C^{(n)}_k,Y_k^{i-1};\!Y_{k,i}\md  {W^{(n)}_{d_{[k-1]}}}, C^{(n)}_{[k-1]},\!Y_{[k-1]}^{i-1}\rp +\! I\lp W^{(n)}_{d_k};C^{(n)}_{[k]}\rp\nonumber\\
    	&\stackrel{\rm (c)}{=}\!
    	\sum_{i=1}^{n}\!I\!\lp\! W^{(n)}_{d_k},\!C^{(n)}_k\!,\!Y_k^{Q-1};\!Y_{k,Q}\!\md\! {W^{(n)}_{d_{[k-1]}}}, \!C^{(n)}_{[k-1]},\!Y_{[k-1]}^{Q-1}, Q\!=\!i\!\rp +\! I\lp W^{(n)}_{d_k};C^{(n)}_{[k]}\rp \nonumber\\
    	& = 
    	nI\hspace{-1pt}\lp W^{(n)}_{d_k},\!C^{(n)}_k,\!Y^{Q-1}_k\!;\! Y_{k,Q}\md {W^{(n)}_{d_{[k-1]}}}, C^{(n)}_{[k-1]},\!Y_{[k-1]}^{Q-1}, Q\rp + I\lp W^{(n)}_{d_k};C^{(n)}_{[k]}\rp\label{eq:nf_3}\\
    	& \stackrel{\rm(d)}{=} nI(U_{k};Y_{k,Q}| U_{k-1})\!+\!I\lp W^{(n)}_{d_k};C^{(n)}_{[k]}\rp\nonumber\\
    	&= nI(U_k;Y_k| U_{k-1})\hspace{-2pt}+\hspace{-2pt}I\lp W^{(n)}_{d_k};C^{(n)}_{[k]}\rp,\nonumber\\
        & \stackrel{\rm(e)}{\leq} nI(U_k;Y_k| U_{k-1})\hspace{-2pt}+\hspace{-2pt}n\mu_{[k]}f(\cache,\bd),\label{eq:nf2_K}
	\end{align}
	where $\rm{(a)}$ holds since, for an 
        uncoded caching strategy and mutually independent files, we have
	\begin{align*}
	I\lp 
         W^{(n)}_{d_k};{W^{(n)}_{d_{[k-1]}}}\md C^{(n)}_{[k]}\rp
         & = I\lp W^{(n)}_{d_k};{W^{(n)}_{d_{[k-1]}}}\md  C_{[k],d_k}, {C^{(n)}_{[k],d_{[k-1]}}}, {C^{(n)}_{[k], [N]\setminus d_{[k]}}}\rp\\
	&=
	H\lp W^{(n)}_{d_k}, C^{(n)}_{[k],d_k}, {C^{(n)}_{[k],d_{[k-1]}}}, {C^{(n)}_{[k], [N]\setminus d_{[k]}}}\rp \\
	&\phantom{=}
	+ H\lp {W^{(n)}_{d_{[k-1]}}},C^{(n)}_{[k],d_k}, {C^{(n)}_{[k],d_{[k-1]}}}, {C^{(n)}_{[k], [N]\setminus d_{[k]}}}\rp\nonumber\\
	&\phantom{=} - H\lp W^{(n)}_{d_k},{W^{(n)}_{d_{[k-1]}}},C^{(n)}_{[k],d_k}, {C^{(n)}_{[k],d_{[k-1]}}}, {C^{(n)}_{[k], [N]\setminus d_{[k]}}}\rp \\
	&\phantom{=}
	- H\lp C^{(n)}_{[k],d_k}, {C^{(n)}_{[k],d_{[k-1]}}}, {C^{(n)}_{[k], [N]\setminus d_{[k]}}}\rp \\
	&=
	H\lp W^{(n)}_{d_k},  {C^{(n)}_{[k],d_{[k-1]}}}, {C^{(n)}_{[k], [N]\setminus d_{[k]}}}\rp \\
	&\phantom{=}
	+ H\lp {W^{(n)}_{d_{[k-1]}}},C^{(n)}_{[k],d_k},  {C^{(n)}_{[k], [N]\setminus d_{[k]}}}\rp \nonumber\\
	&\phantom{=} - H\lp W^{(n)}_{d_k},{W^{(n)}_{d_{[k-1]}}}, {C^{(n)}_{[k], [N]\setminus d_{[k]}}}\rp \\
	&\phantom{=}
	- H\lp C^{(n)}_{[k],d_k}, {C^{(n)}_{[k],d_{[k-1]}}}, {C^{(n)}_{[k], [N]\setminus d_{[k]}}}\rp \\
	&=
	H\lp W^{(n)}_{d_k}\rp + H\lp {C^{(n)}_{[k],d_{[k-1]}}}\rp \\
        &\phantom{=} + H\lp {C^{(n)}_{[k], [N]\setminus d_{[k]}}}\rp + H\lp {W^{(n)}_{d_{[k-1]}}}\rp\\
	&\phantom{=}
	+ H\lp C^{(n)}_{[k],d_k}\rp +H\lp {C^{(n)}_{[k], [N]\setminus d_{[k]}}}\rp \\
	&\phantom{=} - H\lp W^{(n)}_{d_k}\rp -H\lp {W^{(n)}_{d_{[k-1]}}}\rp \\
	&\phantom{=}
	-H\lp {C^{(n)}_{[k], [N]\setminus d_{[k]}}}\rp - H\lp C^{(n)}_{[k],d_k}\rp \\
        & \phantom{=} -H\lp {C^{(n)}_{[k], [N]\setminus d_{[k]}}}\rp\\
	&=0.
	\end{align*}
    Moreover, $\rm{(b)}$ follows from the degradedness of the channel, which implies that for any time instance $i$,  conditioned on $Y_{k,i}$, all channel outputs $\{Y_{u,i}: u<k\}$ are independent of the channel input and hence from the files and cache contents. More precisely, from  ${\lp W^{(n)}_{[N]}, C^{(n)}_{[K]}\rp \!\leftrightarrow \! X_i  \leftrightarrow \!  Y_{K,i} \leftrightarrow \cdots \leftrightarrow Y_{1,i}}$ we have 
    \begin{align*}
	& I\lp W^{(n)}_{d_k};Y_{
	[k-1]}^{i-1}\md {W^{(n)}_{d_{[k-1]}}}, C^{(n)}_{[k]},\!Y^{i-1}_{k}\rp\nonumber\\
	&\!=\! H\lp Y_{[k]}^{i-1}\md  {W^{(n)}_{d_{[k-1]}}}, C^{(n)}_{[k]},\!Y^{i-1}_{k}\rp \! \\
	& \phantom{=} -\! H\lp Y_{[k-1]}^{i-1}\md {W^{(n)}_{d_{[k]}}}, C^{(n)}_{[k]},\!Y^{i-1}_{k}\rp \\
	&= H\lp Y_{[k-1]}^{i-1}\md Y^{i-1}_{k}\rp\hspace{-2pt}-\hspace{-2pt} H\lp Y_{[k-1]}^{i-1}\md Y^{i-1}_{k}\rp\\
	& = 0.
	\end{align*}
	Furthermore, in the equality marked by $\rm{(c)}$, the random variable $Q$ is independent of all other random variables and admits a uniform distribution over $[n]$. In the step $\rm{(d)}$ we have 
        \begin{align*}
            U_k & := \left(U_{k-1},W^{(n)}_{d_k},C^{(n)}_k,Y^{Q-1}_{k}\right) \\
            & = \left( W^{(n)}_{d_{[k]}}, C^{(n)}_{[k]}, Y_{[k]}^{Q-1}, Q\right).
        \end{align*}
	Finally, in the inequality~$\rm{(e)}$ we used~\eqref{eq:I-f}.
	Dividing both sides of~\eqref{eq:nf2_K} by $n$ and letting $n\rightarrow\infty$, we arrive at~\eqref{eq:lm:f_upper_k}, claimed in the lemma. 
	
    Finally, we can use a similar argument for the $K$-th and reach to~\eqref{eq:nf_3}. Continuing from there, we  can write 
	\begin{align}
	nf(\cache,\bd)-n\epsilon_n
        & \leq
	nI\hspace{-1pt}\lp 
        W^{(n)}_{d_K},\!C^{(n)}_K,\!Y^{Q-1}_K\!;\! Y_{K,Q}\md {W^{(n)}_{d_{[K-1]}}}, C^{(n)}_{[K-1]},\!Y_{[K-1]}^{Q-1},Q\rp \nonumber\\
        &\phantom{\leq} + I\lp W^{(n)}_{d_K};C^{(n)}_{[K]}\rp\nonumber\\
	& \stackrel{\rm (a)}{=} nI\hspace{-1pt}\lp \hspace{-1pt}X_Q,W^{(n)}_{d_K},\!C^{(n)}_K,\!Y^{Q-1}_K\!;\! Y_{K,Q}\md  U_{K-1}\hspace{-1pt}\rp \!+\! I\!\lp \hspace{-1pt}W^{(n)}_{d_K};C^{(n)}_{[K]}\rp\nonumber\\
	& \stackrel{\rm (b)}{=} nI\hspace{-1pt}\lp X_Q;\! Y_{K,Q}\md U_{K-1}\rp + I\lp W^{(n)}_{d_K};C^{(n)}_{[K]}\rp \nonumber\\
	& = nI\lp X;Y_K\md U_{K-1} \rp + I\lp W^{(n)}_{d_K};C^{(n)}_{[K]}\rp\nonumber\\ 
	& \stackrel{\rm (c)}{\leq } nI\lp X;Y_K\md U_{K-1} \rp + n\mu_{[K]}f(\cache,\bd),\label{eq:nf_4}
	\end{align}
    where $\rm{(a)}$ holds since the channel input $X_Q$ is deterministically determined by the files and cache contents,~$\rm{(b)}$  
	holds since condition on $X_Q$, the channel output $Y_{K,Q}$ is independent of all other variables, and~$\rm{(c)}$ follows from~\eqref{eq:I-f}. The last inequality in~\eqref{eq:lm:f_upper_K} can be obtained from~\eqref{eq:nf_4}.
	
    It remains to show that the random variables $U_1,\ldots, U_{K-1}$ form a Markov chain. This is immediately implied by the recursive construction of $U_k$ and the fact that $U_{k-1}$ is deterministically known once $U_k$ is given. This completes the proof of the lemma. 
\end{proof}

\begin{proof}[Proof of Lemma~\ref{lm:2user-deg-BC}]
The proof of Lemma~\ref{lm:2user-deg-BC} is derived directly from the proof of Lemma~\ref{lm:Kuser-deg-BC}. From~\eqref{eq:nf1_K} and~\eqref{eq:nf_4}, we have
\begin{align}
    	& nf(\cache,\bd)-n\epsilon_n \leq nI(U_1;Y_1)+n \mu_{\{1\}} f(\cache,\bd),\label{eq:f2-1}\\
    	& nf(\cache,\bd)-n\epsilon_n \leq nI\lp X;Y_2\md U_1 \rp + n \mu_{\{1,2\}} f(\cache,\bd)\label{eq:f2-2},
\end{align}
For the last term in~\eqref{eq:f2-1}, we can write
\begin{align}\label{eq:I2-1}
    \mu_{\{1\}} & = \Big|\bigcup_{\ell\in [\nn_1]} \cI_{1,\ell} \Big| \leq \sum_{\ell\in[\nn_1]} |\cI_{1,\ell}| =  \mu.
\end{align}
Similarly, for the last term in~\eqref{eq:f2-2} we get
\begin{align}\label{eq:I2-2}
     \mu_{\{1,2\}} & = \Big|\bigcup_{u\in \{1,2\}} \bigcup_{\ell\in [\nn_u]} \cI_{u,\ell} \Big|\nonumber\\
     & \leq \sum_{u\in \{1,2\}}\sum_{\ell\in[\nn_u]} |\cI_{u,\ell}|\nonumber\\
     & = \sum_{\ell\in[\nn_1]} |\cI_{1,\ell}| + \sum_{\ell\in[\nn_2]} |\cI_{2,\ell}| = 2\mu.
\end{align}
Plugging~\eqref{eq:I2-1} and~\eqref{eq:I2-2} into~\eqref{eq:f2-1} and~\eqref{eq:f2-2}, respectively, we arrive at the desired inequalities. This completes the proof of the lemma.
\end{proof}

\bibliography{ref_fading}

\end{document}